%% file: main.tex
\pdfoutput=1
\documentclass[11pt]{article}
\input{includes}

\newcommand{\G}{\mathcal{G}}
\newcommand{\I}{\mathcal{I}}
\newcommand{\A}{\mathbb{A}}
\newcommand{\B}{\mathbb{B}}
\newcommand{\C}{\mathbb{C}}
\renewcommand{\S}{\mathbb{S}}
\newcommand{\T}{\mathbb{T}}
\newcommand{\E}{\mathbb{E}}
\DeclareMathOperator{\lw}{LW}

\title{Throttling Equilibria in Auction Markets}
\author{
\sf Xi Chen\footnote{Supported by NSF grants IIS-1838154 and  CCF-1703925.}\\
\sf Columbia University \\
\small\texttt{xichen@cs.columbia.edu}
\and
\sf Christian Kroer\footnote{Supported by ONR awards N00014-22-1-2530 and N00014-23-1-2374, and NSF awards IIS-2147361 and IIS-2238960.} \\
\sf Columbia University\\
\small\texttt{christian.kroer@columbia.edu}
\and
\sf Rachitesh Kumar\\
\sf Carnegie Mellon University \\
\small\texttt{rachitesh@cmu.edu}}
\date{\vspace{-2em}}

\begin{document}

\maketitle

\input{abstract}
\input{intro}
\input{model}

\input{FP-existence-alg}
\input{existence}

\input{hardness}

\input{two_buyer}

\input{TE-vs-PE}

\input{poa.tex}

\input{conclusion}









{
\small
\bibliographystyle{plainnat}
\bibliography{refs}
}

\newpage
\clearpage
\pagenumbering{arabic}
\renewcommand{\thepage}{ec-\arabic{page}}
\begin{center}
    \LARGE 
    Electronic Companion:\\
    Throttling Equilibria in Auction Markets\\
    \large
    
    \author{Xi Chen, Christian Kroer, Rachitesh Kumar}
\end{center}

\appendix

\input{appendix-irrationality}
\input{appendix-PPAD-membership}

\input{appendix-NP-hard}
\input{appendix-two-buyer}
\input{appendix-poa.tex}
\input{appendix-TE-vs-PE}

\end{document}

%% file: includes.tex
\usepackage[dvips, paper=letterpaper, top=0.95in, bottom=.6in, left=0.9in, right=0.95in, nohead, includefoot, footskip=.4in]{geometry}
\usepackage{titlesec}
\titlespacing*{\section} {0pt}{2.5ex plus 1ex minus .2ex}{1.5ex plus .2ex}
\titlespacing*{\subsection} {0pt}{2.25ex plus 1ex minus .2ex}{1ex plus .2ex}
\titlespacing*{\paragraph} {0pt}{2.25ex plus 1ex minus .2ex}{1em}

\usepackage{varioref}
\usepackage[usenames,dvipsnames,xcdraw,table]{xcolor}
\definecolor{DarkBlue}{rgb}{0.1,0.1,0.5}
\definecolor{DarkGreen}{rgb}{0.1,0.5,0.1}
\usepackage{url}
\usepackage[backref=page]{hyperref}
\hypersetup{
     colorlinks   = true,
     linkcolor    = DarkBlue, 
     urlcolor     = DarkBlue, 
	 citecolor    = DarkGreen 
}
\renewcommand*{\backref}[1]{}
\renewcommand*{\backrefalt}[4]{%
    \ifcase #1 (Not cited.)%
    \or        (Cited on page~#2)%
    \else      (Cited on pages~#2)%
    \fi}

\usepackage[authoryear,round]{natbib}
\usepackage{appendix}

\usepackage{mathtools}
\usepackage{mathrsfs}

\usepackage{algorithmic}

\usepackage{algorithm}
\usepackage{array}

\usepackage{subfigure}
\usepackage{subcaption}

\usepackage[font={small,it}]{caption}
\usepackage{setspace}

\usepackage{graphicx,epsfig,amsmath,latexsym,amssymb,verbatim}

\newcommand{\extra}[1]{}

\newcommand{\revision}[1]{{\color{black} #1}}

\usepackage{amsmath,amssymb,amsfonts}
\usepackage{amsthm}
\usepackage{thmtools,thm-restate}
\usepackage{enumitem}

\newtheorem{theorem}{Theorem}

\newtheorem{definition}[theorem]{Definition}
\newtheorem{lemma}[theorem]{Lemma}

\newtheorem{example}{Example}

\def\squareforqed{\hbox{\rlap{$\sqcap$}$\sqcup$}}
\def\qed{\ifmmode\squareforqed\else{\unskip\nobreak\hfil
\penalty50\hskip1em\null\nobreak\hfil\squareforqed
\parfillskip=0pt\finalhyphendemerits=0\endgraf}\fi}
\def\endenv{\ifmmode\;\else{\unskip\nobreak\hfil
\penalty50\hskip1em\null\nobreak\hfil\;
\parfillskip=0pt\finalhyphendemerits=0\endgraf}\fi}

\renewenvironment{proof}{\noindent \textbf{{Proof~} }}{\qed\medskip}
\newenvironment{proof+}[1]{\noindent \textbf{{Proof #1~} }}{\qed\medskip}
\newenvironment{remark}{\noindent \textit{{Remark.~}}}{\qed}

\mathchardef\ordinarycolon\mathcode`\:
\mathcode`\:=\string"8000
\def\vcentcolon{\mathrel{\mathop\ordinarycolon}}
\begingroup \catcode`\:=\active
  \lowercase{\endgroup
  \let :\vcentcolon
  }

\mathchardef\mhyphen="2D

\usepackage{cleveref}

%% file: abstract.tex
\begin{abstract}
Throttling is a popular method of budget management for online ad auctions in which the platform modulates the participation probability of an advertiser in order to smoothly spend her budget across many auctions. In this work, we investigate the setting in which all of the advertisers simultaneously employ throttling to manage their budgets, and we do so for both first-price and second-price auctions. We analyze the structural and computational properties of the resulting equilibria. For first-price auctions, we show that a unique equilibrium always exists, is well-behaved and can be computed efficiently via t\^atonnement-style decentralized dynamics. In contrast, for second-price auctions, we prove that even though an equilibrium always exists, the problem of finding \revision{even an approximate} equilibrium is PPAD-complete, there can be multiple equilibria, and it is NP-hard to find the revenue maximizing one. We also compare the equilibrium outcomes of throttling to those of multiplicative pacing, which is the other most popular and well-studied method of budget management. Finally, we characterize the Price of Anarchy of these equilibria for liquid welfare by showing that it is at most 2 for both first-price and second-price auctions, and demonstrating that our bound is tight.
\end{abstract}

%% file: intro.tex
\section{Introduction}

Online ad auctions are the workhorse of the internet-advertising industry. When a user visits a platform (like Google, Meta, LinkedIn etc.), an auction is run among advertisers who are interested in showing their ad to the user and the winner gets that opportunity. These auctions are run in real time for each user, with the entire process occurring in the milliseconds that the webpage takes to load. Given the heavy user traffic on internet platforms and the low latency with which ad auctions are required to run, it is impossible for any advertiser to manually decide their bid for each user/auction. Instead, platforms provide each advertiser with an automated bidding algorithm that bids on her behalf in the auctions. These bidding algorithms are provided as a (predominantly) free service by the platform, and advertisers can use them by simply specifying their high-level marketing objectives/constraints to the algorithm. Such algorithms offer a clear advantage over the alternative of manually setting bids at the beginning of the advertising campaign and keeping those bids fixed throughout, and thus have been widely adopted by advertisers.

\emph{Budget constraints} are ubiquitous in online advertising, where each advertiser has a maximum amount that they are willing to spend over a given period of time. In light of this ubiquity, one of the primary tasks of automated bidding algorithms is budget management. In this paper, we analyze a popular automated bidding algorithm for budget management known as \emph{throttling} (or alternatively as \emph{probabilistic pacing}), which is used by internet giants like Facebook (Figure~\ref{fig:fb_guide}), Google~\citep{karande2013optimizing} and LinkedIn~\citep{agarwal2014budget}. Throttling manages the expenditure of advertisers by controlling the probability with which they participate in each individual auction. The use of participation probability as a control lever allows the platform to evenly spread out advertisers' expenditure across all auctions in their advertising campaign, while ensuring that they do not spend more than their budgets. This improves the outcomes for both the advertisers and the platform by preventing the former from running out of budget too early, and evenly spreading the competition across all auctions. In the absence of a budget management algorithm to evenly spread the expenditure, there would be heavy competition at the beginning of the campaigns and little to no competition once the advertisers run out of budget. This would hurt advertiser utility because they would pay high prices at the beginning due to increased competition and lose out on lucrative low-competition auctions towards the end of their campaigns. And it would hurt platform revenue because of the lack of competition when the advertisers run out of budget.

Throttling has received significant attention in previous work, the vast majority of which studies it from the perspective of a single buyer participating in repeated (generalized) second-price auctions; see Section~\ref{sec:related-works} for details. In contrast, the scenario where all of the buyers simultaneously employ throttling to manage their budgets, and the resulting system-level properties, has received very little attention. In this paper, we attempt to remedy this by providing a structural and computational analysis of simultaneous multi-buyer throttling for both first-price and second-price auctions. Specifically, we analyze the equilibria that represent the stable points of the system when all buyers simultaneously employ throttling for budget management, with an emphasis on computational properties, convergence rates, efficiency and revenue. 

\begin{figure}
    \centering
    \fbox{\includegraphics[width=0.8\linewidth]{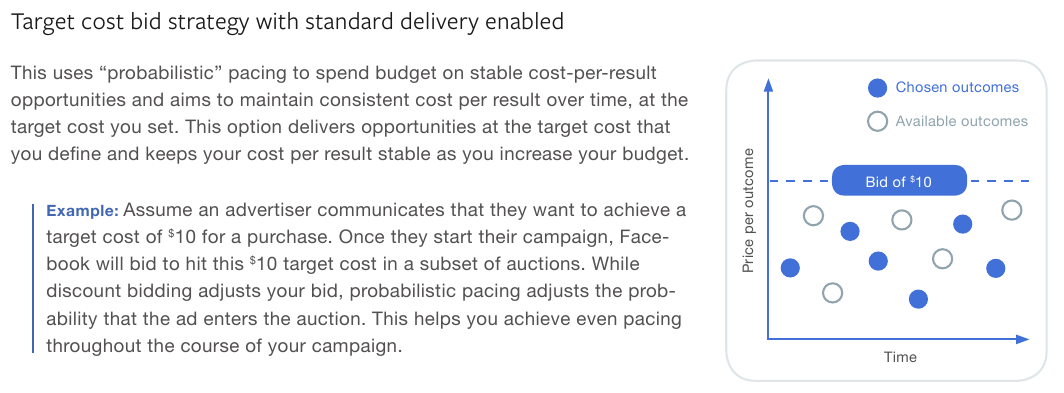}}
    \caption{Excerpt (screenshot) on throttling from the document titled ``Your Guide to Facebook Bid Strategy''~\citep{facebookguide} published by Facebook (now Meta) in September 2017. It refers to throttling by its alternative name ``probabilistic pacing'' and describes it as a budget management strategy that ``adjusts the probability that the ad enters the auction''.}
    \label{fig:fb_guide}
\end{figure}

\subsection{Main Contributions}

We define a \emph{throttling game} with budget-constrained buyers (advertisers) and stochastic good types (user types), in which each buyer chooses the probability with which she participates in the auction.
Repeated play of this throttling game captures the repeated online ad auction setting in which randomly-drawn users arrive on the platform sequentially and each advertiser (or, in practice, the platform-provided algorithm representing the advertiser) employs throttling to smoothly spend their budget over all users. We analyze the Nash equilibria of this throttling game, which we call \emph{throttling equilibria}. They represent the stable points of the system: if all buyers simultaneously employ throttling algorithms to manage their budgets and the algorithms converge, then the point of convergence must be a throttling equilibrium (see Section~\ref{sec:repeated-auctions} for a detailed discussion). We analyze throttling equilibria with an emphasis on its structural and computational properties; a summary of our results follows.

\textbf{First-price Auctions:} We show that a throttling equilibrium always exists, and characterize it as the maximal element in the set of participation probabilities that result in all buyers satisfying their budgets (Theorem~\ref{thm:existence-first-price}). Furthermore, we use this characterization to establish its uniqueness.  On the computational front, we describe decentralized dynamics in which buyers repeatedly play the throttling game and make simple t\^{a}tonnement-style adjustments to their participation probabilities based on their expected expenditure (Algorithm~\ref{alg:dynamics}). We show that these t\^{a}tonnement-style dynamics converge to an approximate throttling equilibrium\footnote{\label{fn:irrational}\revision{Exact throttling equilibria can be irrational even when bids and budgets are integral (in both first- and second-price settings); see Appendix~\ref{appendix:irrationality}. We therefore focus on approximate equilibria and leave open the problem of computing exact equilibria in models of computation that can represent irrational numbers.}} in polynomial time (Theorem~\ref{thm:first-price-dynamics}). In particular, since repeated play of the throttling game captures the setting in which all buyers simultaneously use throttling to manage budgets, our algorithm specifies a method for iteratively updating participation probabilities of individual advertisers which efficiently converges to an approximate throttling equilibrium when employed by every advertiser. Finally, in Section~\ref{sec:numerics}, we numerically analyze throttling equilibria in first-price auctions using a semi-synthetic dataset.

\textbf{Second-price Auctions:} We begin by establishing that a throttling equilibrium always exists for second-price auctions (Theorem~\ref{theorem:existence_second_price}), but find that it may not be unique, and for some games all throttling equilibria can be irrational. Next, we prove results about the computational complexity of finding  approximate throttling equilibria\footref{fn:irrational}. Before summarizing those results, we make a note for readers who may not be familiar with complexity theory: In order to make our results more accessible, we provide an informal description at the head of every subsection to avoid complexity-theoretic terminology from obscuring the conclusions derived from them; please refer to these descriptions for a less technical exposition. 
Continuing on with the summary of our results, we prove that the problem of computing approximate throttling equilibria is PPAD-hard even when 
each good has at most three bids (Theorem~\ref{thm:second-price-PPAD-hard}), by showing a reduction from the PPAD-hard problem of computing an approximate equilibrium of a threshold game~\citep{PB2021}. As a consequence, we show that, unlike first-price auctions, no dynamics can converge in polynomial time to a second-price throttling equilibrium (assuming PPAD-complete problems cannot be solved in polynomial time). In other words, no matter which throttling algorithm the platform designs to evenly spend the budget of each advertiser, it will not always converge to an equilibrium efficiently, i.e., the participation probabilities will not always stabilize. Furthermore, to complete the picture, we place the problem of computing approximate throttling equilibria  in the class PPAD by showing a reduction to the problem of finding a Brouwer fixed point 
(Theorem~\ref{thm:second-price-PPAD-membership}); the latter is known to be in PPAD via Sperner's lemma (e.g. see \citealt{chen2009settling}). We provide additional evidence of the computational challenges that afflict throttling for second-price auctions by proving the NP-hardness of finding a revenue-maximizing approximate throttling equilibrium (Theorem~\ref{thm:revenue_np_hard}). We complement these hardness results by describing a polynomial-time algorithm for computing throttling equilibria for the special case in which there are at most two non-zero bids on each good (Algorithm~\ref{alg:two_buyer}), thereby precisely delineating the boundary of computational tractability.

\textbf{Price of Anarchy:}  Liquid welfare \citep{dobzinski2014efficiency, azar2017liquid} is a measure of efficiency for settings with budget constraints, like the one considered in this work. It corresponds to the maximum revenue (liquidity) that can be extracted with full knowledge of the buyers' values/bids\footnote{Since we make no assumptions about the values of the buyers in our model, we define liquid welfare as the maximum revenue that can be extracted given fixed bids and budgets. We do so for ease of exposition and our definition reduces to the traditional definition of liquid welfare based on values if one assumes truthful bidding.}, and reduces to social welfare when the budgets are non-binding. We show that the liquid welfare under any throttling equilibrium is at most a factor of 2 away from the liquid welfare that can be obtained by a central planner with complete information of the buyers bids/values, i.e., the Price of Anarchy is at most 2. We do so for both first-price and second-price auctions. Moreover, we provide examples to show that this bound is tight for both auction formats.

\subsection{Comparison of Pacing and Throttling}\label{sec:intro-pacing-throttling}

Multiplicative pacing is another popular method of budget management offered by platforms. Instead of controlling the participation probability \'{a} la throttling, it multiplicatively shades the bid of the buyer and attempts to evenly spend the budget across all auctions by controlling the amount of shading. Unlike throttling, the dynamics and equilibria that result from the simultaneous use of pacing have received significant attention in the literature~\citep{borgs2007dynamics,balseiro2015repeated, balseiro2017budget, conitzer2018multiplicative, conitzer2019pacing, chen2021complexity, balseiro2021contextual, gaitonde2022budget}. Our comprehensive analysis of throttling allows us to extensively compare the equilibrium properties of throttling with those established for pacing in prior work; we discuss the similarities and differences here.

Despite being very different algorithms that manage budgets using distinct controls (participation probabilities vs bid shading), we find that pacing and throttling have remarkably similar computational and structural properties at equilibrium, as summarized in Table~\ref{table:first_price} and Table~\ref{table:second_price}. In light of this comparison, our work can be seen as providing the analogous set of results for throttling equilibria that \citet{borgs2007dynamics, conitzer2018multiplicative, conitzer2019pacing, chen2021complexity} proved for pacing equilibria. Even though throttling and pacing share similar computational and structural properties, the techniques required to establish these properties are nonetheless very distinct. In particular, our main results about efficient convergence of simple dynamics in first-price auctions and the impossibility of convergence of any dynamics (PPAD hardness) in second-price auctions rely on completely novel proof techniques that share little in common with the work on pacing. 
Intuitively, one can see why this would be the case because of the potentially irrational nature of throttling equilibria, whereas pacing equilibria have a nice convex program characterization for first-price auctions and a nice mixed-integer program characterization for second-price auctions.
This makes the similarity in results all the more surprising, and provides independent evidence for what the analysis of pacing suggested: budget management for first-price auctions is more well-behaved as compared to second-price auctions.

\begin{table}[H]
\scriptsize
\centering
\begin{tabular}{|l|l|l|l|l|}
\hline
\textbf{Existence}                                   & \textbf{Rationality}                                 & \textbf{Multiplicity}                                & \textbf{Computational Complexity}                    & \textbf{Efficient Dynamics}                        \\ \hline
Always                                      & Not always                                  & Always unique                               & Poly.-time for approx. eq.                  & For approx. eq.                           \\ \hline
Always                                      & Always                                      & Always unique                               & Poly.-time for exact eq.                    & For approx. eq.                           \\
\citep{conitzer2019pacing} & \citep{conitzer2019pacing} & \citep{conitzer2019pacing} & \citep{conitzer2019pacing} & \citep{borgs2007dynamics} \\ \hline
\end{tabular}

\caption{Comparison of throttling (top row) and pacing equilibria (bottom row) for \textbf{first-price} auctions.}
\label{table:first_price}
\end{table}

\begin{table}[H]
\scriptsize
\centering
\begin{tabular}{|l|l|l|l|l|}
\hline
\textbf{Existence} & \textbf{Rationality} & \textbf{Multiplicity} & \textbf{Computational Complexity}                     & \textbf{Revenue Max.} \\ \hline
Always    & Not always  & Possibly infinite     & PPAD-complete for approx. eq                 & NP-hard      \\ \hline
Always    & Always      & Possible     & PPAD-complete for both exact & NP-hard      \\
\citep{conitzer2018multiplicative}     & \citep{chen2021complexity}        & \citep{conitzer2018multiplicative}        & and approx. eq. \citep{chen2021complexity}                                         & \citep{conitzer2018multiplicative}        \\ \hline
\end{tabular}
\caption{Comparison of throttling equilibria (top row) and pacing equilibria (bottom row) for \textbf{second-price} auctions.}
\label{table:second_price}
\end{table}

In first-price auctions, the uniqueness of throttling and pacing equilibrium allows us to compare the platform revenue under the two budget management methods: we show that the revenue under throttling and pacing is always within a factor of 2 of each other (Theorem~\ref{thm:revenue-comparison}). However, neither method dominates the other in terms of revenue under first-price auctions---we give examples where throttling out-earns pacing and vice versa. In our numerical experiments on semi-synthetic data, we find that throttling yields slightly more revenue than pacing in first-price auctions; see Section~\ref{sec:numerics} for details. For second-price auctions, the multiplicity of equilibria makes a direct revenue comparison difficult. 
For the special case of symmetric and independent buyer values, \citet{balseiro2017budget} showed that throttling yields greater revenue than pacing in equilibrium. Thus, a revenue-maximizing platform would prefer to offer throttling over pacing in that scenario.

The impact on the utility of the advertisers also warrants attention. If one fixes the competing bids (or their distribution), then it has been shown for a multitude of settings that pacing yields a best response in second-price auctions~\citep{balseiro2019learning,conitzer2018multiplicative}. In other words, pacing is the optimal method of budget management for an individual advertiser in second-price auctions when the competing bids are stationary. However, competing bids are not stationary when all of the advertisers also simultaneously use pacing to manage their budgets and update the amount by which they shade their bids. In fact, the PPAD-hardness result~\citep{chen2021complexity} implies that the system is not guaranteed to stabilize in that setting, and in particular, the competing bids will not always become stationary with time and can continue to change. This voids the best-response guarantees of pacing in the setting where all advertisers simultaneously use it for bidding in second-price auctions. In the special case where values are symmetrically distributed, \citet{balseiro2017budget} found that pacing yields greater utility for the advertisers than throttling in equilibrium. More generally however, if one considers asymmetric values, we provide examples in Section~\ref{sec:TE-PE} which show that either pacing or throttling can yield higher utility for \emph{all} advertisers in second-price auctions.
Utility comparisons also prove difficult in first-price auctions. They are impossible to make without making stringent, and often unrealistic, assumptions on how the buyers determine their bids based on their values. Unlike second-price auctions, where bidding truthfully is a dominant strategy in the absence of budgets, dominant strategies don't exist in first-price auctions and even Bayes-Nash equilibrium strategies are not well-understood in asymmetric environments. 

We conclude by mentioning some practical advantages of throttling. In contrast to other budget-management methods like pacing, throttling does not modify the bids of the advertisers. This makes throttling ideal for advertisers who do not opt into budget-management services that change their bids---the platform can only control the participation probability of such advertisers, thereby making throttling the natural choice~\citep{karande2013optimizing}. Moreover, throttling also gives advertisers a more representative sample of users for which they are eligible and their bid is competitive~\citep{karande2013optimizing}; see Section~\ref{sec:numerics} for numerical evidence on realistic data. 
This is in contrast to pacing, which biases the allocation towards users where the advertiser has a high \emph{estimated} probability of getting a click (referred to as click-through-rate), relative to other advertisers. The ability of throttling to gather a more representative sample of data on clicks enables better estimates of click-through-rates in dynamic or under-explored environments. More generally, \citet{gui2022auction} formally showed that throttling is conducive to estimating the causal effect of advertising campaigns---a very valuable property in the age of AI-generated ads. Many advertisers place a premium on the ability to gather representative samples of users and estimate the causal effect of their ad campaign, motivating platforms to offer throttling as a budget-management option. Finally, for first-price auctions, throttling also ensures that the cost-per-click remains stable, i.e., the payment for each click that the advertiser wins remains the same throughout the campaign. This is because the payment of the winning advertiser is equal to her bid in first-price auctions, and throttling does not modify bids. The ability to control and predict the cost-per-click is often valuable in practice, making throttling a desirable choice for advertisers (see Figure~\ref{fig:fb_guide}).

\subsection{Additional Related Work}\label{sec:related-works}

Budget management in online ad auctions has received widespread attention in the literature. Here, we review the papers which most closely relate to ours. We begin by reviewing the work on throttling, then give a broad overview of the other models for bidding under budget constraints considered in the literature, and finally conclude with a discussion of the work on multiplicative pacing, which we use in our comparisons.

Among work on throttling,
\citet{balseiro2017budget} is the closest to ours, and acts as the inspiration for our model and terminology. In \citet{balseiro2017budget}, the authors study system equilibria under various budget management strategies in second-price auctions, one of them being throttling. They show that throttling satisfies desirable incentive properties, and, under the special case when buyer values are symmetric, they compare equilibrium seller revenue and total welfare for throttling to a variety of other budget management techniques. Crucially, unlike our work, \citet{balseiro2017budget} lacks a computational analysis and focuses solely on second-price auctions. Throttling has also received significant attention in other lines of research. \citet{agarwal2014budget} study throttling in generalized second-price (GSP) auctions from the perspective of a single buyer, provide an algorithm which determines the participation probability based on user traffic forecasts, and analyze its performance empirically on real data from LinkedIn. Similarly, \citet{xu2015smart} provide and empirically evaluate practical algorithms for throttling on data from demand-side platforms. \citet{karande2013optimizing} use throttling (under the name Vanilla Probabilistic Throttling) as the benchmark in the GSP auction setting to evaluate the budget management algorithm they describe on data from Google, and find that it empirically outperforms throttling on the metrics they study. Importantly, they do not engage in an equilibrium analysis, and their algorithm does not provide a representative sample of the traffic to advertisers. Finally, \citet{gui2022auction} recently showed that the randomness in throttling can be used to estimate the treatment effect of advertising campaigns and validated their findings on data from a real-world platform that uses throttling.

There is also a significant body of work which proposes alternatives to pacing and throttling.
\citet{charles2013budget} study regret-free budget-smoothing policies in which the platform selects the random subset of buyers that participate in the GSP auction for each good. They show that such policies always exist, and, under the small-bids assumption, give an efficient algorithm for the special case of second-price auctions.  
There is also a significant body of work that models budget management in first-price auctions as an allocation problem under stochastic and adversarial input, starting with the work of \citet{mehta2007adwords}. See \citet{mehta2013online} for a survey.
We emphasize that in this work, our goal is to study the budget-management mechanisms that are employed in practice, which motivates our focus on throttling and its comparison with pacing.

Unlike throttling, equilibrium analysis for the setting where all buyers simultaneously use multiplicative pacing is well-understood. \citet{conitzer2018multiplicative} define and study pacing equilibria in second-price auctions. They show that it always exists and study its structural properties. \citet{chen2021complexity} recently analyzed the computational complexity of finding pacing equilibria in second-price auctions, and proved that the problem is PPAD-complete. For first-price auctions, \citet{borgs2007dynamics} describe a simple t\^{a}tonnement-style dynamics and prove its efficient convergence to a pacing equilibrium. \citet{conitzer2019pacing} characterize the structural properties of pacing equilibria in first-price auctions and give a market-equilibrium based algorithm for computing it. A summary of these results can be found in the second row of Table~\ref{table:first_price} and Table~\ref{table:second_price}. \revision{Like much of the literature on pacing equilibria, we take buyers' bids and budgets as exogenous inputs and focus on the equilibrium outcomes induced by the platform's budget-management layer.  This modeling choice is motivated by the common separation between ML-based bid/value estimation and budget control in practice, and by the fact that the computational and structural properties of these equilibria remain important regardless of whether higher-level budget/bid-cap parameters are themselves strategically chosen.  \citet{feng2024strategic} is a notable exception: they study the Price of Anarchy in a metagame where buyers strategically choose budgets and maximum bids; their analysis crucially relies on the assumption that the platform implements the First-Price Pacing Equilibrium and builds upon the computational/structural results established in \citet{conitzer2019pacing}.}


Finally, Price of Anarchy for liquid welfare has been studied in a variety of settings with budget-constrained buyers. \citet{azar2017liquid} show that the Price of Anarchy of liquid welfare is 2 for simultaneous first-price and second-price auctions in the multi-item setting. For mixed-strategy Nash equilibria they showed an upper bound of  51.5, with both results requiring the ``no over-budgeting" assumption that requires the sum of bids to be less than the budget. Their results cannot be compared to ours because we only require the budget constraints to be satisfied in expectation, and more generally the settings have different action sets (throttling parameters versus bids). \citet{balseiro2021contextual} study a contextual-value auction model with strategic agents and establish the existence of pacing-based equilibria for all standard auction formats, including first-price and second-price auctions as special cases. They go on to show that pacing equilibria have a Price of Anarchy of at most 2 for all standard auctions in their model. \citet{gaitonde2022budget} study the repeated-auction setting and show that the Price of Anarchy is bounded above by 2 even if the system is not in equilibrium, provided that all of the buyers use a generalization of the dual-descent-based pacing algorithm introduced in \citet{balseiro2019learning}.

%% file: model.tex
\section{Model}

Consider a seller (platform) who has $m$ types of goods to sell, and $n$ budget constrained buyers (advertisers) who are interested in buying these goods. The seller runs an auction amongst the buyers in order to make the sale. We assume that the type of good to be sold is drawn from some known distribution $d = (d_1, \dots, d_m)$, i.e., the good to be sold is of type $j$ with probability $d_j$. Buyer $i$ bids $\tilde{b}_{ij}$ on good type $j$, for all $i \in [n], j \in [m]$, and has a per-auction budget of $B_i > 0$. In practice, this bid is often computed as $\tilde b_{ij} = \bar b_i \cdot CTR_{ij}$, where $\bar b_i$ is the maximum bid-per-click specified by the advertiser and $CTR_{ij}$ is the user-specific click-through-rate estimated by the platform.
To control their budget expenditure, each buyer $i$ is associated with a \emph{throttling parameter} $\theta_i \in [0,1]$, which represents the probability with which she participates in the auction: each buyer $i$ independently flips a biased coin which comes up heads with probability $\theta_i$, and submits their bid $\tilde{b}_{ij}$ if the coin comes up heads, while submitting no bid if the coin comes up tails. 

We focus on the setting where each buyer wishes to satisfy her budget constraint in expectation, i.e., buyer $i$ would like to spend less than $B_i$ in expectation over the good types and participation coin flips of all buyers. Concentration arguments can be used to show that in-expectation budget constraints for individual auctions appropriately capture global budget constraints over a large number of repeated auctions; see Section~\ref{sec:repeated-auctions} for a detailed discussion. For this reason, in-expectation budget constraints have been used extensively in prior work~\cite{gummadi2012repeated,abhishek2013optimal,balseiro2015repeated,balseiro2017budget,balseiro2019learning,conitzer2018multiplicative} to model online-advertising auctions. 
Additionally, in this paper, we restrict our attention to the two most commonly-used auction formats in online advertising: first-price auctions and second-price auctions. In a first-price auction, the participating buyer with the highest bid wins the good and pays her bid, whereas in a second-price auction, the participating buyer with the highest bid wins the good and pays the second-highest bid among the participating buyers. Our model  can be interpreted as a discrete version of the one defined in \cite{balseiro2017budget}.

Before proceeding further, we introduce some additional notation that allows us to capture the
stochastic nature of the good types via a rescaling of the bids, thereby allowing us to analyze the setting as a deterministic multi-good auction problem:
Set $b_{ij} \coloneqq d_j \tilde{b}_{ij}$ for all $i \in [n], j \in [m]$. Since the participation of buyers is independent of the good type and we are only concerned with expected payments, the good type distribution $d = (d_j)_j$ and the bids $\{\tilde{b}_{ij}\}_{i,j}$ are consequential only insofar as they determine $\{b_{ij}\}_{i,j}$. Therefore, with some abuse of terminology, going forward, we refer to $b_{ij}$ as the bid of buyer $i$ on good $j$ (instead of $\tilde{b}_{ij}$, which will no longer be used).\footnote{This deterministic view is equivalent to the  model of \citet{conitzer2018multiplicative}, except that we focus on probabilistic throttling for managing budgets, whereas they focus on multiplicative pacing. Their model can similarly be viewed as a stochastic setting.}
Furthermore, to simplify our analysis, we will assume that ties are broken lexicographically, i.e., the smaller buyer number wins in case of a tie. Our results continue to hold for all other tie-breaking priority orders over the buyers (even when they are different for each good). The lexicographic tie-breaking rule allows for simplified notation, albeit with some abuse: We will write $b_{ij} > b_{kj}$ to mean that either $b_{ij}$ is strictly greater than $b_{kj}$, or $b_{ij} = b_{kj}$ and $i<k$. Finally, we refer to any tuple $\left(n, m, (b_{ij}), (B_i)\right)$ 
as a \emph{throttling game}. 

In online-advertising auctions, the buyers (or, more typically in practice, the platform on behalf of the buyers) attempt to satisfy their budget constraints by adjusting their throttling parameters. This naturally leads to a game where each buyer's strategy is her throttling parameter. We use $p(\theta)_{ij}$ to denote the expected payment of buyer $i$ on good $j$ when buyers use $\theta = (\theta_1, \dots, \theta_n)$ to decide their participation probabilities. Let $X_i$ be a random variable such that $X_i = 1$ if buyer $i$ participates and $X_i = 0$ if buyer $i$ does not participate. Then, by our modeling assumptions, $X_i$ is a Bernoulli$(\theta_i)$ random variable. More concretely, $p(\theta)_{ij}$ can be defined as follows:
\begin{itemize}
	\item First-price auction: $p(\theta)_{ij} = \E \left[X_i b_{ij} \prod_{k: b_{kj} > b_{ij}} (1 - X_k)  \right] = \theta_i b_{ij} \prod_{k: b_{kj} > b_{ij}} (1 - \theta_k)$
	\item Second-price auction:
	\begin{align*}
	    p(\theta)_{ij} &= \E \left[ \sum_{\ell: b_{\ell j} < b_{ij}}  b_{\ell j} X_i  X_\ell  \prod_{k \neq i: b_{kj} > b_{\ell j}} (1 - X_k) \right] 
	     = \sum_{\ell: b_{\ell j} < b_{ij}}  b_{\ell j} \theta_i  \theta_\ell  \prod_{k \neq i: b_{kj} > b_{\ell j}} (1 - \theta_k)
	\end{align*}
\end{itemize}

We overload $p(\theta)_{ij}$ to represent the expected payment in both auction formats; the auction format will be clear from the context. We assume here that the participation probability of a buyer across goods is perfectly correlated for simplicity ($X_i$ is the same for all $j$). Any other correlation structure, e.g. independent across goods, would also lead to the same results due to linearity of expectation. Next, we define the equilibrium concept which will be the main object of study in this work.

\begin{definition}[Throttling Equilibrium]\label{def:exact_throt_eq}
	Given a throttling game $\left(n, m, (b_{ij}), (B_i)\right)$, a vector of throttling parameters $\theta = (\theta_1, \dots, \theta_n) \in [0,1]^n$ is called a throttling equilibrium if:
    \begin{enumerate}
        \item Budget constraints are satisfied: $\sum_j p(\theta)_{ij} \leq B_i$ for all $i \in [n]$
        \item No unnecessary throttling occurs: If $\sum_j p(\theta)_{ij} < B_i$, then $\theta_i = 1$
    \end{enumerate}
\end{definition}

The above definition applies to both first-price and second-price auctions using the corresponding payment rule $p(\theta)_{ij}$. Throttling equilibria are closely connected to pure Nash Equilibria of the throttling game---they are essentially identical in a natural utility model. Consider a throttling game $\left(n, m, (b_{ij}), (B_i)\right)$. Fix an auction format: either first-price or second-price. Suppose buyer $i$ has value $v_{ij}$ for good type $j$ for all $i \in [n], j \in [m]$. We make the natural assumption that the buyers bid less than their value, i.e., $d_i v_{ij} \geq b_{ij}$ for all $i \in [n], j \in [m]$ in second-price auctions and strictly less than their value, i.e., $d_i v_{ij} > b_{ij}$ for all $i \in [n], j \in [m]$ in first-price auctions. Define a new game $G$ in which each buyer $i$'s strategy is her throttling parameter $\theta_i$ and her utility function is given by 
\begin{align*}
	u_i(\theta) = \begin{cases}
		\sum_j \left(v_{ij} d_i  \theta_i \prod_{k: b_{kj} > b_{ij}} (1 - \theta_k)  - p(\theta)_{ij} \right) & \text{if } \sum_j p(\theta)_{ij} \leq B_i\\
		-\infty &\text{otherwise}
	\end{cases}	
\end{align*}
This utility is simply the expected quasi-linear utility modified to ascribe a utility value of negative infinity for budget violations. Since all of the buyers get a non-negative utility from winning any good, increasing the throttling parameter improves utility so long as the budget constraint is satisfied. This makes it easy to see that every throttling equilibrium of the throttling game $\left(n, m, (b_{ij}), (B_i)\right)$ is a Nash equilibrium of the corresponding game $G$. Furthermore, it is also straightforward to check that a pure Nash equilibrium of $G$ is not a throttling equilibrium only in the following scenario: There is a buyer who spends 0 in the Nash equilibrium and has a throttling parameter strictly less than 1. In this scenario, setting the throttling parameter of all such buyers to 1 yields a throttling equilibrium with exactly the same expected allocation and payment for all the buyers as under the Nash equilibrium. Hence, given a throttling game $\left(n, m, (b_{ij}), (B_i)\right)$, for every pure Nash equilibrium of the corresponding game $G$, there is a throttling equilibrium with the same expected allocation and revenue.

We next define an approximate version of throttling equilibrium, which allows us to side-step issues of irrationality that can plague exact equilibria (see Example~\ref{example:first-price-irrational} and Example~\ref{example:irrational_eq}).

\begin{definition}[Approximate Throttling Equilibrium]\label{def:throt_eq}
    Given a throttling game $\left(n, m, (b_{ij}), (B_i)\right)$, a vector of throttling parameters $\theta = (\theta_1, \dots, \theta_n)$ is called an $\delta$-approximate throttling equilibrium if:
    \begin{enumerate}
        \item Budget constraints are satisfied: $\sum_j p(\theta)_{ij} \leq B_i$ for all $i \in [n]$
        \item Not too much unnecessary throttling occurs: If $\sum_j p(\theta)_{ij} < (1- \delta) B_i$, then $\theta_i \geq 1- \delta$
    \end{enumerate}
\end{definition}

\subsection{Connection to Throttling in Repeated Auctions}\label{sec:repeated-auctions}

Before proceeding onto our results, we motivate our model by connecting it to throttling in repeated auctions. Consider $n$ buyers who participate in $T$ sequential auctions with the help of automated bidding algorithms. Let $\tilde B_i$ denote the total budget of buyer $i$---it is the maximum amount she is willing to spend across all $T$ auctions. Like before, we will assume that the type of good in each auction is drawn from the distribution $d = (d_1, \dots, d_m)$ and buyer $i$ bids $\tilde b_{ij}$ on good type $j$. Each buyer employs throttling for budget management, i.e., a throttling-based bidding algorithm maintains a throttling parameter $\theta(t)_i$ for each buyer $i$ and participates on her behalf in auction $t$ with probability $\theta(t)_i$.

Suppose the algorithms use equilibrium throttling parameters $\theta(t)_i = \theta^*_i$ for some throttling equilibrium $(\theta_1^*, \dots, \theta_n^*)$, then each buyer would spend $B_i = \tilde B_t/T$ in expectation in each auction or have $\theta(t)_i = 1$. If $T$ is large, as is the case in practice, then standard concentration arguments imply that the cumulative expenditure of each buyer $i$ is close to her total budget $\tilde B_i = B_t \cdot T$ with high probability. More precisely, if $\theta^*_i < 1$, then the total expenditure of buyer $i$ would be $\tilde B_i \pm O(\sqrt{T})$ with high probability. In other words, no buyer would run out of budget before the last $O(\sqrt{T})$ auctions, and each buyer would either have $\theta_i = 1$ or spend at least $\tilde B_i - O(\sqrt{T})$ in total. Moreover, the total expenditure of each buyer is evenly spread out across all auctions (on average), as desired. The aforementioned concentration argument motivates the use of in-expectation budget constraints to model high volume online ad auctions, both here and in prior work~\citep{gummadi2012repeated,abhishek2013optimal,balseiro2015repeated,balseiro2017budget,balseiro2019learning,conitzer2018multiplicative}.

In practice, the system does not always start out or stay exactly at equilibrium. Instead, the throttling parameters are iteratively adjusted at a regular cadence to achieve an even spread of expenditure over all auctions: if the expenditure in auctions $t$ through $t+k-1$ is less than $B_i\cdot k$, then the throttling parameter is increased to $\theta(t+k)_i > \theta(t)_i$, otherwise it is decreased to $\theta(t+k)_i < \theta(t)_i$, where $k$ determines the update frequency. Typically these changes are iterative with a small step size $\eta \geq |\theta(t+k)_i - \theta(t)_i|$. We refer the reader to the algorithm of \cite{agarwal2014budget} for a concrete example of a practical implementation of such an algorithm. Suppose such an iterative algorithm converges to some profile of throttling parameters $\theta^* = (\theta^*_1, \dots, \theta^*_n)$, i.e., the iterates stay in a small neighbourhood $(1 - \epsilon) \cdot \theta^* \leq \theta(t) \leq \theta^*$ for all auctions $t \geq t^*$ after some time $t^*$. Such a stable point $\theta^*$ has to be an approximate throttling equilibrium. To see this, observe that for each buyer $i \in [n]$ we have
\begin{enumerate}
    \item Budget constraints are satisfied: Expected payment satisfies $\sum_j p(\theta^*)_{ij} \leq B_i/(1-\epsilon)^n$. This is because otherwise buyer $i$ would spend $\sum_j p(\theta(t))_{ij} > B_i$ on average, which in turn would cause her throttling parameter to decrease and become less than $(1 - \epsilon) \cdot \theta^*_i$ over time, thereby contradicting our assumption of convergence to $\theta^*_i$.
    \item Not too much unnecessary throttling: Either expected payment $\sum_j p(\theta^*)_{ij} \geq (1 -\epsilon)^n \cdot B_i$ or $\theta^*_i \geq 1 - \epsilon$. For contradiction, suppose $\theta^*_i < 1 - \epsilon$ and $\sum_j p(\theta^*)_{ij} < (1 -\epsilon)^n \cdot B_i$. Then, buyer $i$ would spend $\sum_j p(\theta(t))_{ij} < B_i$ on average, which in turn would cause an increase in her throttling parameter, eventually increasing beyond $\theta^*_i$ and leaving the convergence neighbourhood, thereby yielding the desired contradiction.
\end{enumerate}

Therefore, $\theta^*$ is a $\delta$-approximate throttling equilibrium for the throttling game with slightly modified budgets $B_i' = B_i/(1-\delta)$, where $\delta$ is chosen to be small enough to satisfy $(1 - \delta) \leq (1 - \epsilon)^{2n}$. Thus, any stable point to which the system converges under simultaneous use of throttling algorithms has to be an approximate throttling equilibrium; this fact motivates our focus on (approximate) throttling equilibria. 

\revision{Recall that throttling equilibria correspond to pure-strategy Nash equilibria for the game $G$ in which the strategy of the buyer is her throttling parameter. The aforementioned argument establishes the equivalence of \emph{last-iterate-convergence} points of t\^atonnement-based dynamics and approximate throttling equilibria. We show that such a convergence always occurs for first-price auctions when throttling parameters are updated using \Cref{alg:dynamics}, and conversely, for second-price auctions, it cannot always occur in polynomial time if PPAD$\neq$P. This impossibility of last-iterate convergence in polynomial time is also the case in general-sum normal form games, which has in part motivated the study of more general notions of equilibria like Correlated and Coarse Correlated equilibria to capture the \emph{joint empirical outcome distribution} of the dynamics induced by no-swap-regret and no-external-regret algorithms respectively. Analogously, one can consider the joint empirical distribution of the iterates produced by no-external-regret (or no-swap-regret) algorithms in the game $G$ and analyze the corresponding notion of Coarse Correlated Equilibria (or Correlated Equilibria) for throttling. However, the ``convergence" to Coarse Correlated Equilibria does not imply stabilization of the iterates, i.e., last-iterate convergence), which is the desideratum that motivates throttling equilibria in  our discussion. 
The stabilization implied by last-iterate convergence is very attractive in large-scale systems such as the ones we model, for example because: it has improved stability properties when it comes to experimentation~\citep{kroer2023fair,liao2025interference}; its (near)-optimality can be easily observed and explained to stakeholders; it exposes cycling and divergent behavior; it allows for faster convergence on ``benign'' instances than the worst-case theory predicts---whereas average convergence is largely constrained to the rate implied by the averaging scheme. The game-solving literature has a long line of works studying last-iterate convergence exactly for these reasons~\citep{anagnostides2022last,cai2022finite,cai2024fast,daskalakis2017training}.
See also \citet{papadimitriou2019game} for a comprehensive argument advocating the use of stable points/sets as preferable solution concepts.
}



%% file: FP-existence-alg.tex
\section{Throttling in First-price Auctions}

We begin by studying throttling equilibria in the first-price auction setting. We start by showing that, for first-price auctions, there always exists a unique throttling equilibrium. We then describe a simple and efficient t\^{a}tonnement-style algorithm for computing approximate throttling equilibria.

\subsection{Existence of First-Price Throttling Equilibria}

To show existence, we will characterize the throttling equilibrium as a component-wise maximum of the set of all budget-feasible throttling parameters. This argument is inspired by the technique used in \citet{conitzer2019pacing} for pacing equilibria in first-price auctions. We use the following definition to make the argument precise. 

\begin{definition}[Budget-feasible Throttling Parameters]
	Given a throttling game $\left(n, m, (b_{ij}), (B_i)\right)$, a vector of throttling parameters $\theta \in [0,1]^n$ is called budget-feasible if every buyer satisfies her budget constraints, i.e., $\sum_j p(\theta)_{ij} \leq B_i$ for all buyers $i \in [n]$.
\end{definition}

The following lemma captures the crucial fact that the component-wise maximum of two budget-feasible throttling parameters is also budget-feasible.

\begin{lemma}\label{lemma:pairwise_max}
	Given a throttling game $\left(n, m, (b_{ij}), (B_i)\right)$ for first-price auction, if $\theta, \tilde{\theta} \in [0,1]^n$ are budget-feasible, then $\max(\theta, \tilde{\theta}) \coloneqq (\max(\theta_i, \tilde{\theta}_i))_i$ is also budget-feasible.
\end{lemma}

\begin{proof}
	Fix $i \in [n]$ and $j \in [m]$. Without loss of generality, we assume that $\theta_i \geq \tilde{\theta}_i$. Observe that
	\begin{align*}
		p(\max(\theta, \tilde{\theta}))_{ij} 	= \prod_{k: b_{kj} > b_{ij}} (1 - \max(\theta_k, \tilde{\theta}_k)) \theta_i b_{ij} \leq \prod_{k: b_{kj} > b_{ij}} (1 - \theta_k) \theta_i b_{ij} = p(\theta)_{ij}
	\end{align*}
	Summing over all goods $j \in [m]$ completes the proof.
\end{proof}

The maximality property shown in \cref{lemma:pairwise_max} is analogous to the maximality property of multiplicative pacing: there it is also the case that component-wise maxima over pacing vectors preserves budget feasibility for first-price auctions, and this was used by \citet{conitzer2019pacing} to show several structural properties of pacing equilibria.
Next we show that the maximality property allows us to establish the existence and uniqueness of throttling equilibria for first-price auction.

\begin{theorem}\label{thm:existence-first-price}
	For every throttling game $\left(n, m, (b_{ij}), (B_i)\right)$ for first-price auction, there exists a unique throttling equilibrium $\theta^* \in [0,1]^n$ which is given by
 $$
		\theta^*_i = \sup\{\theta_i \in [0,1] \mid \exists\ \theta_{-i} \text{ such that } \theta = (\theta_i, \theta_{-i}) \text{ is budget-feasible}\}.$$ 
\end{theorem}

\begin{proof}
	Set $\theta^*_i = \sup\{\theta_i \in [0,1] \mid \exists\ \theta_{-i} \text{ such that } \theta = (\theta_i, \theta_{-i}) \text{ is budget-feasible}\}$ for all $i \in [n]$. First, we show that $\theta_i^*$ is budget-feasible. Observe that the function $\theta \mapsto \left(\sum_j p(\theta)_{1j}, \dots, \sum_j p(\theta)_{nj} \right)$ is continuous. Therefore, the pre-image of the set $\bigtimes_{i=1}^n[0,B_i]$ under this function is closed. In other words, the set of all budget-feasible throttling parameters is closed. Fix $\epsilon > 0$. For all $i \in [n]$, by the definition of $\theta^*_i$, there exists $\theta^{(i)} \in [0,1]^n$ which is budget feasible and $\theta^{(i)}_i > \theta_i^* - \epsilon$. Iterative application of Lemma~\ref{lemma:pairwise_max} yields the budget-feasibility of the vector $\theta$ defined by $\theta_i = \max_{k \in [n]} \theta^{(k)}_i$. Moreover, $\theta_i > \theta^*_i - \epsilon$ for all $i \in [n]$ because $\theta^{(i)}_i > \theta_i^* - \epsilon$ for all $i \in [n]$. Since $\epsilon> 0$ was arbitrary, we have shown that there exists a sequence of budget-feasible throttling parameters which converges to $\theta^*$, which implies that $\theta^*$ is budget-feasible because the set of budget-feasible throttling parameters is closed.
	
	Next, we show that $\theta^*$ also satisfies the `No unnecessary pacing' property. Suppose there exists $i \in [n]$ such that $\sum_j p(\theta^*)_{ij} < B_i$ and $\theta^*_i < 1$. Then, by the continuity of $\theta \mapsto \sum_j p(\theta)_{ij}$, there exists $\theta_i$ such that $\theta^*_i < \theta_i< 1$ and $\sum_j p(\theta_i, \theta^*_{-i})_{ij} < B_i$, which contradicts the definition of $\theta^*$. Therefore, for all $i \in [n]$,  we have $\sum_j p(\theta^*)_{ij} < B_i$ implies $\theta^*_i = 1$. Thus, $\theta^*$ is a throttling equilibrium.
	
	Finally, we prove uniqueness of $\theta^*$. Suppose there is a throttling equilibrium $\theta \in [0,1]^n$ such that $\theta_i \neq \theta^*_i$ for some $i \in [n]$. Then, the set of buyers $C \subset [n]$ for whom $\theta_i^* > \theta_i$ is non-empty. Note that $\theta_i < 1$ for all $i \in C$. Hence, every buyer in $C$ spends her entire budget under $\theta$. On the other hand, since $\theta_i^* > \theta_i$ for all $i \in C$ and $\theta_i^* = \theta_i$ for $i\notin C$, the total payment made by buyers in $C$ under $\theta^*$ is strictly greater than the payment under $\theta$, which contradicts the budget-feasibility of $\theta^*$. Therefore, $\theta^*$ is the unique throttling equilibrium.
\end{proof}

\revision{In Appendix~\ref{appendix:irrationality}, we describe a throttling game for first-price auction with integer bids and budgets where the unique throttling equilibrium has irrational throttling parameters for some buyers. In other words, rational throttling equilibrium need not always exist. Since irrational numbers cannot be represented exactly with a finite number of bits in the standard floating point representation, it leads us to consider algorithms for finding \emph{approximate} throttling equilibrium instead. We leave open the question of characterizing the complexity of computing the exact throttling equilibrium in computation models or complexity classes which can handle irrational real numbers (for examples, see \citealt{blum2012complexity} and \citealt{etessami2010complexity}).}

\subsection{An Algorithm for Computing Approximate First-Price Throttling Equilibria}

In this subsection, we define a simple t\^atonnement-style algorithm and prove that it yields an approximate throttling equilibrium in polynomial time.

\begin{algorithm}[H] 
   \caption{Dynamics for First-price Auction}
   \label{alg:dynamics}
    \begin{algorithmic}\vspace{0.08cm}
            \item[\textbf{Input:}] Throttling game $\left(n, m, (b_{ij}), (B_i)\right)$ and approximation parameter $\delta \in (0,1/2)$
            \item[\textbf{Initialize:}] $\theta_i = \min\{B_i/ (2\sum_j b_{ij}), 1\}$ for all $i \in [n]$
            \item[\textbf{While}] there exists a buyer $i \in [n]$
            such that $\theta_i < 1 - \delta$ and $\sum_j p(\theta)_{ij} < (1 - \delta) B_i$:
            \begin{itemize}
                \item For all $i \in [n]$ such that $\theta_i < 1 - \delta$ and $\sum_j p(\theta)_{ij} < (1 - \delta) B_i$, set $\theta_i \leftarrow \theta_i/(1 - \delta)$ 
            \end{itemize}
            \item[\textbf{Return:}] $\theta$
    \end{algorithmic}
\end{algorithm}

Before proceeding to prove the correctness and efficiency of Algorithm~\ref{alg:dynamics}, we note some of its properties. Typically, in online advertising auctions, buyers participate in a large number of auctions throughout their campaign, and the platform periodically updates their throttling parameters to ensure that they don't finish their budgets prematurely and lose out on valuable advertising opportunity; see Section~\ref{sec:repeated-auctions} for a discussion. The above algorithm is especially suited for such a setting due to its decentralized and easy-to-implement updates to the throttling parameter---each buyer can independently and iteratively update their throttling parameter using just their own expected payment. In particular, if each buyer employs the simple automated bidding algorithm that increases their throttling parameter $\theta_i$ by a multiplicative factor $1/(1-\delta)$ whenever $\theta_i < 1 -\delta$ and $\sum_j p(\theta)_{ij} < (1 - \delta) B_i$, then the system would behave like Algorithm~\ref{alg:dynamics}, and therefore converge to an equilibrium in polynomial time. This fact also lends credence to the notion of throttling equilibrium as a solution concept because it predicts the outcome of the system when buyers employ such simple iterative throttling algorithms.


We prove the efficient convergence of Algorithm~\ref{alg:dynamics} in the remainder of this section. To do so we first show that, throughout the run of Algorithm~\ref{alg:dynamics}, all the buyers always satisfy their budget constraints.

\begin{lemma}\label{lemma:no_budget_violation}
	Consider the run of Algorithm~\ref{alg:dynamics} on the throttling game $\left(n, m, (b_{ij}), (B_i)\right)$ and approximation parameter $\delta \in (0,1/2)$. Then, after every iteration of the while loop, we have $\sum_{j} p(\theta)_{ij} \leq B_i$ for all $i \in [n]$. 
\end{lemma}
\begin{proof}
	We prove the lemma using induction on the number of iterations of the while loop. Note that $\sum_{j} p(\theta)_{ij} \leq B_i$ for all $i \in [n]$ before the first iteration of the while loop by virtue of our initialization. Let $\theta$ and $\theta'$ represent the throttling parameters after the $t$-th iteration and the $(t+1)$-th iteration of the while loop. Suppose $\sum_{j} p(\theta)_{ij} \leq B_i$ for all $i \in [n]$. To complete the proof by induction, we need to show that $\sum_{j} p(\theta')_{ij} \leq B_i$ for all $i \in [n]$. Consider a buyer $i$. Suppose $\sum_{j} p(\theta)_{ij} \geq (1 - \delta) B_i$. By the update step of the algorithm, $\theta'_i = \theta_i$ and $\theta'_j \geq \theta_j$ for $j \neq i$. Therefore,
	\begin{align*}
 		\sum_{j} p(\theta')_{ij} = 	\sum_j \prod_{k: b_{kj} > b_{ij}} (1 - \theta'_k) \theta'_i b_{ij} \leq 	\sum_j \prod_{k: b_{kj} > b_{ij}} (1 - \theta_k) \theta_i b_{ij} = \sum_{j} p(\theta)_{ij} \leq B_i
	\end{align*}
	On the other hand, suppose $\sum_{j} p(\theta)_{ij} < (1 - \delta) B_i$. Then, by the update step of the algorithm, $\theta'_i \leq \theta_i/(1- \delta)$ and $\theta'_j \geq \theta_j$ for $j \neq i$. Therefore,
	\begin{align*}
		\sum_{j} p(\theta')_{ij} = 	\sum_j \prod_{k: b_{kj} > b_{ij}} (1 - \theta'_k) \theta'_i b_{ij} \leq \frac{1}{1-\delta} \cdot\sum_j \prod_{k: b_{kj} > b_{ij}} (1 - \theta_k) \theta_i b_{ij} = \frac{1}{1-\delta} \cdot\sum_{j} p(\theta)_{ij} < B_i
	\end{align*}
	This completes the proof of $\sum_{j} p(\theta')_{ij} \leq B_i$ for all buyers $i \in [n]$.
\end{proof}

We conclude this subsection by proving the correctness and efficiency of Algorithm~\ref{alg:dynamics}.

\begin{theorem}\label{thm:first-price-dynamics}
	Given a throttling game $\left(n, m, (b_{ij}), (B_i)\right)$ for first-price auction and an approximation parameter $\delta \in (0,1/2)$ as input, Algorithm~\ref{alg:dynamics} returns a $\delta$-approximate throttling equilibrium in polynomial time. 
\end{theorem}

\begin{proof}
	Since each iteration of the while loop only performs basic arithmetic operations, to establish a polynomial run-time complexity, it suffices to show that the while loop terminates in polynomially-many steps. Note that $c = \min_i \min\{B_i/ (2\sum_j b_{ij}), 1\}$ is a lower bound on the initial value of the throttling parameter of every buyer. Due to the termination condition of the while loop and the update step, at each iteration of the while loop, there exists a buyer $i \in [n]$ whose throttling parameter is updated, i.e., there exists $i \in [n]$ such that $\theta_i < 1- \delta$ and $\theta_i \leftarrow \theta_i/(1 - \delta)$. Therefore, the number of iteration of the while loop $T$ satisfies the following relationships:
	\begin{align*}
		\frac{c}{(1 -\delta)^{T/n}} \leq 1 \iff T \leq \frac{n \log(1/c)}{\log(1/ (1-\delta))} \leq \frac{n\log(1/c)}{\delta}
	\end{align*}
	The second sequence of inequalities upper bounds the number of iterations of the while loop by a polynomial function of the problem size.
	
	To complete the proof, it suffices to show that at the termination of the while loop, $\theta$ is a $\delta$-approximate throttling equilibrium. Budget-feasibility follows  from Lemma~\ref{lemma:no_budget_violation}, and the `Not too much unnecessary throttling' condition is satisfied by virtue of the termination condition. 
\end{proof}

%% file: existence.tex
\section{Throttling in Second-price Auctions}

In this section, we study throttling equilibria in second-price auctions. We begin by establishing the existence of exact throttling equilibria. Next, we show that, in contrast to first-price auctions, it is PPAD-hard to compute an approximate throttling equilibrium. To complete the characterization of its complexity, we then place the problem in PPAD. Moreover, we also show that, unlike first-price auctions, multiple throttling-equilibria can exist for second-price auctions, and finding the revenue-maximizing one is NP-hard. Finally, we complement these negative results with an efficient algorithm for the case when each good has at most two positive bids.

\subsection{Existence of Second-Price Throttling Equilibria}

The following theorem establishes the existence of an exact throttling equilibrium for every bid profile by invoking Brouwer's fixed-point theorem on an appropriately defined function.

\begin{theorem}\label{theorem:existence_second_price}
	For second-price auctions, there exists a throttling equilibrium $\theta^* \in [0,1]^n$ in every throttling game $\left(n, m, (b_{ij}), (B_i)\right)$.
\end{theorem}

\begin{proof}
	First, observe that we can write the expected payment of buyer $i$ on good $j$ under $\theta$ as 
	\begin{align}\label{eq:exp_pay}
		p(\theta)_{ij} = \sum_{\ell: b_{\ell j} < b_{ij}}  b_{\ell j} \theta_i  \theta_\ell  \prod_{k \neq i: b_{kj} > b_{\ell j}} (1 - \theta_k) = \theta_i \cdot \sum_{\ell: b_{\ell j} < b_{ij}}  b_{\ell j}  \theta_\ell  \prod_{k \neq i: b_{kj} > b_{\ell j}} (1 - \theta_k)	 = \theta_i \cdot p(1, \theta_{-i})_{ij}
	\end{align}

	Define $f: [0,1]^n \to [0,1]^n$ as
	\begin{align*}
    	f_i(\theta) = \min\left\{ \frac{B_i}{\sum_j p(1, \theta_{-i})_{ij}}, 1 \right\} \quad \forall \theta \in [0,1]^n
	\end{align*}
	where, we assume that $f_i(\theta) = 1$ if $\sum_j p(1, \theta_{-i})_{ij} = 0$. Note that $p(1, \theta_{-i})_{ij}$ is a continuous function of $\theta$ because it is a polynomial. Therefore, $f$ is continuous as a function of $\theta$ and hence, by Brouwer's fixed-point theorem, there exists a $\theta^*$ such that $f(\theta^*) = \theta^*$. As a consequence, for all buyers $i \in [n]$, we get the following equivalent statements
	\begin{align*}
		f_i(\theta^*) = \theta^*_i \iff \theta^*_i \cdot \sum_j p(1, \theta^*_{-i})_{ij} < B_i \text{ implies } \theta^*_i = 1 \iff \sum_j p(\theta^*)_{ij} < B_i \text{ implies } \theta^*_i = 1
	\end{align*}
	where the last equivalence follows from equation~\ref{eq:exp_pay}. Moreover, by definition of $f$, we get
	\begin{align*}
		\theta^*_i = f_i(\theta^*) \leq 	\frac{B_i}{\sum_j p(1, \theta_{-i})_{ij}}
	\end{align*}
	which in conjunction with \eqref{eq:exp_pay} implies $\sum_j p(\theta^*)_{ij} \leq B_i$. Therefore,  $\theta^*$ is a throttling equilibrium.
\end{proof}

\begin{remark}
    The fixed-point argument of Theorem~\ref{theorem:existence_second_price} can also be used to show the existence of a throttling equilibrium in first-price auctions. In fact, it extends to any auction format in which $p(1, \theta_{-i})_{ij}$ is a continuous function, e.g., convex combination of first-price and second-price, running first-price auction for some items and second-price for others etc. \revision{Moreover, as the functions $f_i(\theta) = B_i/\max\{B_i, \sum_j p(1, \theta_{-i})_{ij}\}$ are composed of the gates $\{+, -, *, /, \max, \min\}$, the problem of computing an exact throttling equilibrium lies in the class FIXP~\citep{etessami2010complexity} for all such auction formats. The notion of approximate equilibrium that we study (Definition~\ref{def:throt_eq}) is a ``weak approximation"~\citep{etessami2010complexity} because it corresponds to an approximate fixed point and need not be close in distance to an exact fixed point; see Appendix~\ref{appendix:irrationality} for examples. We leave open the question of $\text{FIXP}_a$ hardness, i.e., the complexity of computing an approximate throttling equilibrium which is guaranteed to be $\delta$-close in $\ell_\infty$ to an exact throttling equilibrium.}
\end{remark}

Even though the above theorem establishes the existence of a throttling equilibrium for every throttling game, in Appendix~\ref{appendix:irrationality} we give an example of a throttling game for second-price auction in which all equilibria have a buyer with an irrational throttling parameter. This prompts us to study the problem of computing \emph{approximate} throttling equilibria in the following subsections.

%% file: hardness.tex
\subsection{Complexity of Finding Approximate Second-Price Throttling Equilibria}\label{sec:complexity_spte}
 
In the previous subsection, by way of our existence proof, we reduced the problem of finding an approximate throttling equilibrium to that of finding a Brouwer fixed point of the function $f$; but this is of little use if we want to actually compute an approximate throttling equilibrium: no known algorithm can compute a Brouwer fixed point in polynomial time and it is believed to be a hard problem. This is because the problem of computing an approximate Brouwer fixed point is a complete problem for the class PPAD \citep{papadimitriou1994complexity}. Informally, this means that it is as hard as any other problem in the class, such as computing Nash equilibria of bimatrix games~\citep{daskalakis2009complexity, chen2009settling} or computing a market equilibrium under piece-wise linear concave utilities~\citep{chen2009spending, vazirani2011market}, all of which have eluded a polynomial-time algorithm despite decades of intensive effort.

However, through our reduction, we have only shown that the problem of computing an approximate throttling equilibrium is easier than the problem of computing a Brouwer fixed point by showing that any algorithm for the latter can be employed to solve the former. Perhaps, computing an approximate throttling equilibrium is strictly easier? Unfortunately, this is not the case and the goal of this subsection is to prove it. More precisely, we show that the problem of finding an approximate throttling equilibrium is PPAD-hard, which in informal terms means that it is as hard as any other problem in the class PPAD, in particular that of computing a Brouwer fixed point. Before stating the hardness result itself, we note a consequence of particular importance: Under the assumption that PPAD-hard problems cannot be solved in polynomial time, no dynamics can efficiently converge to an approximate throttling equilibrium in polynomial time, which is in stark contrast to throttling in first-price auctions. In other words, no matter which throttling algorithm is used to evenly spread the expenditure, the system is not guaranteed to stabilize (as defined in Section~\ref{sec:repeated-auctions}) in polynomial time. Now, we state the main result of the section.
\color{black}
\begin{theorem}\label{thm:second-price-PPAD-hard}
There is a positive constant $\delta<1$ such that the problem
  of finding a $\delta$-approximate throttling equilibrium in a throttling game for second-price auction is PPAD-hard. 
This holds even when the number of buyers with non-zero bids 
  for each good is at most three.
\end{theorem}

  The proof of Theorem~\ref{thm:second-price-PPAD-hard} uses \emph{threshold games}, introduced recently by \citet{PB2021}.
They showed that the problem of finding an approximate equilibrium in a threshold game is PPAD-complete. %
\begin{definition}[Threshold game of \citealt{PB2021}]
	\label{def:threshold}
	A threshold game is defined over a directed graph $\G = ([n], E)$. Each node $i\in [n]$ represents a player with strategy space $x_i\in [0,1]$. 
	Let $N_i$ be the set of nodes~$j\in [n]$ with $(j,i)\in E$.
Then $ (x_i:i\in [n]) \in [0, 1]^{n}$ is an $\epsilon$-\emph{approximate equilibrium} if every $x_i$ satisfies
%
	\begin{align*}
	x_i \in  
	\begin{cases}
	[0,\epsilon] & \sum_{j \in N_i} x_j > 0.5 + \epsilon\\
	[1-\epsilon,1] & \sum_{j \in N_i} x_j < 0.5 - \epsilon\\
	[0,1] & \sum_{j \in N_i}x_j \in[ 0.5 - \epsilon, 0.5+\epsilon] 
	\end{cases} 
	\end{align*}
\end{definition}

\begin{theorem}[Theorem 4.7 of \citealt{PB2021}]
	\label{thm:ppad-threshold}
	There is a positive constant $\epsilon<1$ such that the problem of 
  finding an $\epsilon$-approximate equilibrium in a threshold game is PPAD-hard. 
This holds even when the in-degree and out-degree of each node is at most three in the threshold game.
\end{theorem}

Given a threshold game $\G=([n],E)$, we let $O_i$ denote the set of nodes $j\in V$ with $(i,j)\in E$. 
So $|N_i|,|O_i|\le 3$ for every $i\in [n]$. To prove Theorem~\ref{thm:second-price-PPAD-hard}, we need to construct a throttling game $\I_\G$ from $\G$ such that any approximate throttling equilibrium of $\I_\G$ corresponds to an approximate equilibrium of the threshold game. Before rigorously diving into the construction, we give an informal description to build intuition.

 {With each node of $\G$, we will associate a collection of buyers and goods, with the goal of capturing the corresponding strategy and equilibrium conditions of the threshold game. Fix a node $i \in [n]$. We will define a strategy buyer $S(i)$ and set the strategy $x_i$ for node $i$ to be proportional to $1 - \theta_{S(i)}$. Next, in order to implement the equilibrium condition of the threshold game, our goal will be to define buyers and goods such that the linear form $\sum_{j \in N_i} x_j$ ends up being proportional to the total payment of a buyer who we will refer to as the threshold buyer $T(i)$. For each in-neighbour $j \in N_i$, we will define a neighbour good $G(i,j)$, for which the strategy buyer of the neighbour $S(j)$ places the highest bid of $6$ and the threshold buyer $T(i)$ places a bid of 5. Furthermore, for a reason that will become clear shortly, the strategy buyer $S(i)$ places a bid of 4 on $G(i,j)$. With these bids, the payment made by the threshold buyer $T(i)$ on all the neighbour goods $\{G(i,j)\}_j$ is proportional to $\theta_{T(i)} \theta_{S(i)} \sum_{j \in N_i} (1 - \theta_S(j)) = \theta_{T(i)} \theta_{S(i)} \sum_{j \in N_i} x_j$. Now, if we are somehow able to ensure that the throttling parameter of the strategy buyer $\theta_{S(i)}$ is inversely proportional to the throttling parameter of the threshold buyer $\theta_{T(i)}$, then the payment made by the threshold buyer on all the neighbour goods will be proportional to $\sum_{j \in N_i} x_j$, as desired. To achieve this, we define a reciprocal good $R(i)$. Finally, we set the budget of the threshold buyer $T(i)$ in such a way that comparing it to her payment, which is proportional to $\sum_{j \in N_i} x_j$, is tantamount to making a comparison between $\sum_{j \in N_i} x_j$ and $0.5$. The challenging part of the reduction lies in setting up the bids and budgets in a way that ensures that this comparison leads to an enforcement of the equilibrium condition of the threshold game.}

\begin{figure}
    \centering
    \includegraphics[width = 3.5in, angle = 0]{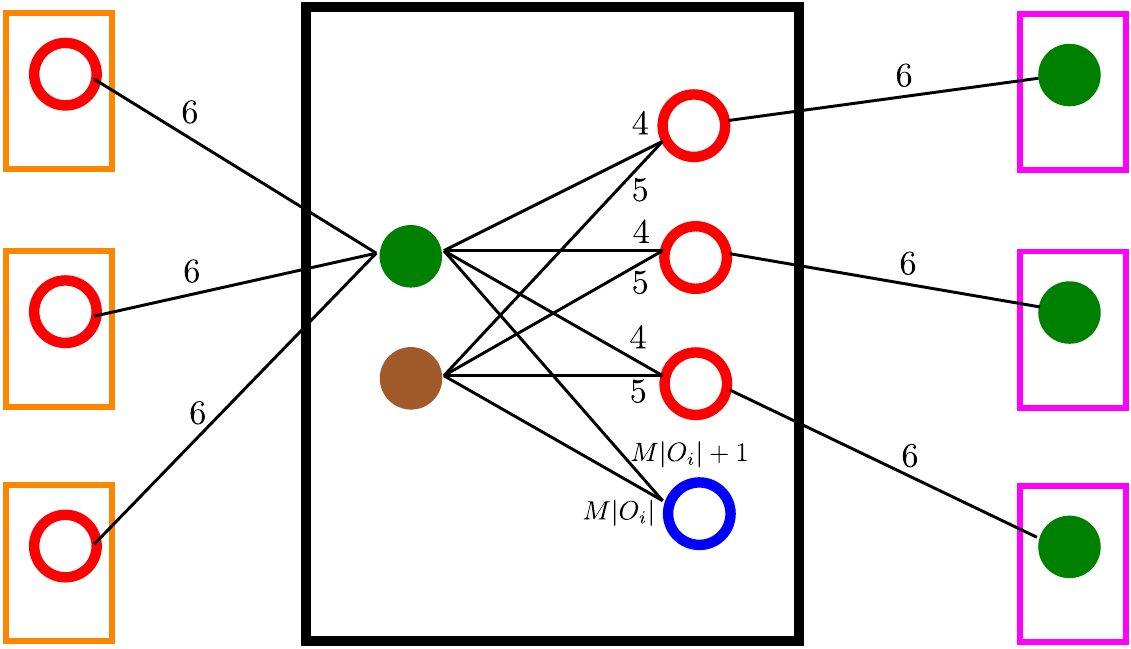}
    \caption{A diagrammatic representation of the non-zero bids made and received by the buyers and goods corresponding to a particular node of the threshold game. Consider a particular node of the threshold game, represented here by the black rectangle. It has three \textcolor{orange}{out neighbours}, depicted as orange rectangles, and three \textcolor{magenta}{in neighbours}, depicted as pink rectangles. The hollow circles represent goods and the solid circles represent buyers. Corresponding to each node of the threshold game, there is one \textcolor{OliveGreen}{strategic buyer}, shown here in green, and one \textcolor{brown}{threshold buyer}, shown in brown. Furthermore, for each node of the threshold game, there are three \textcolor{red}{neighbour goods}, shown in red, and one \textcolor{blue}{reciprocal good}, shown in blue. The lines represent non-zero bids.}
    \label{fig:PPAD}
\end{figure}

With this high level overview of the reduction in place, we move on to the rigorous construction of the throttling game $\I_\G$ from $\G$. First, $\I_\G$ contains the following set of goods:
\begin{itemize}
    \item For each $i \in [n]$ and $j \in N_i$, there is a \emph{neighbor good} $G(i,j)$.
    \item For each $i \in [n]$, there is a \emph{reciprocal good} $R(i)$.
\end{itemize}
Next, setting two constants $M$ and $\delta$ as $ M = {160}/{\delta}$ and 
$\delta=\min  \{\epsilon/(3 + \epsilon), \epsilon/2, 1/4  \},$
the throttling game $\I_\G$ has the following set of buyers:
\begin{flushleft}\begin{itemize}
    \item For each $i\in [n]$, there is a \emph{threshold buyer} $T(i)$ who has
    budget $1/2$ and has non-zero bids only on the following goods:
    $b(T(i), G(i,j)) = 5$ for all $j \in N_i$; and
     $b(T(i), R(i)) = M|O_i|$.
   
    \item For each $i \in [n]$, there is a \emph{strategy buyer} $S(i)$
    who has budget $M|O_i|/2$ and has non-zero bids only on the following goods:
     $b(S(i), R(i)) = M|O_i| + 1$;
      $b(S(i), G(i,j)) = 4$ for all $j \in N_i$; and
       $b(S(i), G(j, i)) = 6$ for all $j \in O_i$.
\end{itemize}\end{flushleft}
It is clear that $\I_\G$ can be constructed from $\G$ in polynomial time, and the number of buyers with non-zero bids
for each good is at most three.
Let $\theta$  be any $\delta$-approximate throttling equilibrium of $\I_\G$ 
and use it to define $(x_i:i\in [n])\in [0,1]^n$ as follows: 
\begin{equation}\label{heheeq}
x_i = \min\big\{2(1 - \theta_{S(i)}),1\big\},\quad\text{for all $i \in [n]$.}
\end{equation}
To complete the reduction, we show that $  (x_i:i\in [n])$ is an $\epsilon$-approximate equilibrium of the threshold game $\G$. 

\begin{lemma}\label{lemma:hardness}
 $(x_i:i\in [n])$, as defined in (\ref{heheeq}), is an $\epsilon$-approximate equilibrium of the threshold game $\G$.
\end{lemma}

Lemma~\ref{lemma:hardness} completes the reduction, and thereby the proof of Theorem~\ref{thm:second-price-PPAD-hard}, because we have shown that for any $\delta$-approximate throttling equilibrium of the throttling game $\I_\G$, the strategy $(x_i)_i$ is an $\epsilon$-approximate equilibrium of the threshold game $\G$. The proof of Lemma~\ref{lemma:hardness} is the technical crux of the reduction; it can be found in Appendix~\ref{appendix:hardness}.

\paragraph{PPAD Membership of Approximate Second-Price Throttling Equilibria}

Next, we show that the problem of computing a $\delta$-approximate throttling equilibrium belongs to PPAD by showing a reduction to BROUWER: the problem of computing an approximate fixed point of a Lipschitz continuous function from a $n$-dimensional unit cube to itself, which known to be in PPAD~\citep{chen2009settling}. Its proof is motivated by the argument for existence of exact throttling equilibria given in Theorem~\ref{theorem:existence_second_price} and can be found in Appendix~\ref{appendix:PPAD-membership}

\begin{theorem}\label{thm:second-price-PPAD-membership}
	For second-price auction, the problem of computing an approximate throttling equilibrium is in PPAD.
\end{theorem}

\subsection{NP-hardness of Revenue Maximization under Throttling}

To further strengthen our hardness result, next we establish the NP-hardness of computing the revenue-maximizing approximate throttling equilibrium. With revenue being one of the primary concerns of advertising platforms, this result provides further evidence of the computational difficulties which plague throttling equilibria in second-price auctions. We begin by defining the decision version of the revenue-maximization problem.

\begin{definition}[REV]
    Given a throttling game $G$ and target revenue $R$ as input, decide if there exists a $\delta$-approximate throttling equilibrium of $G$, for any $\delta \in [0,1)$, with revenue greater than or equal to $R$.
\end{definition}

Note that we allow for arbitrarily bad approximations to the throttling equilibrium by allowing $\delta$ to be any number in $[0,1)$. Theorem~\ref{thm:revenue_np_hard} states the problem of finding the revenue maximizing approximate throttling equilibrium is NP-hard. Its based on a reduction from 3-SAT to REV and has been relegated to Appendix~\ref{appendix:NP-hard}.
\begin{theorem}\label{thm:revenue_np_hard}
	REV is NP-hard for second-price auction.
\end{theorem}

%% file: two_buyer.tex
\subsection{An Algorithm for Second-Price Throttling Equilibria with Two Buyers Per Good}

Next, we contrast the hardness result of the previous subsection with an algorithm for the case when each good receives at most two non-zero bids. Since goods with only one positive bid never result in a payment, without loss of generality, we can assume that every good has exactly two buyers with positive bids. More precisely, in this subsection, we will assume that $|\{i \in [n] \mid b_{ij}>0\}| = 2$ for all $j \in [m]$. This special case demarcates the boundary of tractability for computing throttling equilibria in second-price auctions: Our PPAD-hardness result (Theorem~\ref{thm:second-price-PPAD-hard}) holds for the slightly more general case of each good receiving at most three positive bids. We begin by describing the algorithm (Algorithm~\ref{alg:two_buyer}).

\begin{algorithm}[H] 
   \caption{Algorithm for the Two Buyer Case}
   \label{alg:two_buyer}
    \begin{algorithmic}\vspace{0.08cm}
            \item[\textbf{Input:}] Throttling game $\left(n, m, (b_{ij}), (B_i)\right)$ and parameter $\gamma > 0$
            \item[\textbf{Initialize:}] $\theta_i = \min\{B_i/ (2\sum_j b_{ij}), 1\}$ for all $i \in [n]$
            \item[\textbf{While}] there exists a buyer $i \in [n]$
            such that $\theta_i < 1 - \gamma$ and $\sum_j p(\theta)_{ij} < (1 - \gamma)^3 B_i$:
            \begin{enumerate}
                \item For all $i \in [n]$ such that $\theta_i < 1 - \gamma$ and $\sum_j p(\theta)_{ij} < (1 - \gamma)^2 B_i$, set $\theta_i \leftarrow \theta_i/(1 - \gamma)$
                \item For all $i \in [n]$ such that $\sum_j p(\theta)_{ij} > B_i$, set $\theta_i \leftarrow (1 - \gamma)\theta_i$ \vspace{0.1cm}  
            \end{enumerate}
            \item[\textbf{Return:}] $\theta$
    \end{algorithmic}
 \end{algorithm}

The following theorem, whose proof can be found in Appendix~\ref{appendix:two-buyer}, establishes the correctness and polynomial runtime of Algorithm~\ref{alg:two_buyer}.

\begin{theorem}\label{thm:two-buyer}
    Algorithm~\ref{alg:two_buyer} returns a $3 \gamma$-approximate throttling equilibrium in time which is polynomial in the size of the instance and $1/\gamma$.
\end{theorem}

%% file: TE-vs-PE.tex
\section{Comparing Pacing and Throttling}\label{sec:TE-PE}

In this section, we compare two of the most popular budget management strategies: multiplicative pacing and throttling. First, we restate the definition of pacing equilibrium, as it appears in \citet{conitzer2018multiplicative, conitzer2019pacing}. Under pacing, each buyer $i$ has a pacing parameter $\alpha_i$ and, she bids $\alpha_i b_{ij}$ on good $j$. Let $p_j(\alpha)$ denote the price on good $j$ when all of the buyers use pacing, i.e., $p_j(\alpha)$ is the highest (second-highest) element among $\{\alpha_i b_{ij}\}_i$ for first-price (second-price) auctions. Then, a tuple $((\alpha_i), (x_{ij}))$ of pacing parameters and allocations $x_{ij}$ is called a pacing equilibrium if the following hold:
\begin{itemize}[noitemsep]
    \item[(a)] Only buyers with the highest bid win the good: $x_{ij} > 0$ implies $\alpha_i b_{ij}= \max_i \alpha_i b_{ij}$.\vspace{0.1cm}
    \item[(b)] Full allocation of each good with a positive bid:  $\max_i \alpha_i b_{ij} > 0$  implies $\sum_{i\in [n]} x_{ij} = 1$.\vspace{0.1cm}
    \item[(c)] Budgets are satisfied: $\sum_{j\in [m]} x_{ij} p_j(\alpha) \leq B_i$.\vspace{0.1cm}
	\item[(d)] 
	No unnecessary pacing: $\sum_{j\in [m]} x_{ij} p_j(\alpha) < B_i$ implies $\alpha_i=1$.\vspace{0.15cm}
\end{itemize}

The following example illustrates the concept of pacing equilibrium in second-price auctions, and gives an example where there are infinitely-many throttling equilibria, which will be discussed when we compare pacing and throttling in second-price auctions.
\begin{example}
    \label{ex: inf many spte}
    There are 2 goods and 2 buyers. The bids are given by $b_{11} = b_{22} = 2$, $b_{12} = b_{21} = 1$, and the budgets are given by $B_1 = B_2 = 1/2$. Then, it is straightforward to check that any pair of throttling parameters $\theta_1, \theta_2 \in [1/2, 1]$ such that $\theta_1 \theta_2 = 1/2$ forms a throttling equilibrium. 

    There is a symmetric pacing equilibrium where $\alpha_1 = \alpha_2 = 1/2$, $x_{11} = x_{22} = 1$ and $x_{12} = x_{21} = 0$.
    However, there are also two highly asymmetric equilibria: let $\alpha_1 = 2/3$ and $\alpha_2 = 1/3$. Then $x_{11}=1$ and the price is $1/3$ for the first good. The bids are tied on the second good, and the price is $2/3$. We set $x_{22} = 3/4$ such that buyer 2 spends their whole budget, and $x_{12}= 1/4$, such that buyer 1 spends their whole budget. 
    The other asymmetric equilibrium is obtained by swapping the roles of the buyers.
\end{example}

\paragraph{Comparing Pacing and Throttling in First-Price Auctions}

We begin with a comparison of pacing equilibria and throttling equilibria for first-price auctions. In \citet{conitzer2019pacing}, the authors show that a unique pacing equilibrium always exists in first-price auctions and characterize it as the largest element in the collection of all budget-feasible vectors of pacing parameters. In Theorem~\ref{thm:existence-first-price}, we show the analogous result for throttling using similar techniques. However, unlike pacing equilibrium, which is known to be rational \citet{conitzer2019pacing}, there exist throttling games where the throttling equilibrium is irrational as we demonstrate through Example~\ref{example:first-price-irrational} in Appendix~\ref{appendix:irrationality}. Furthermore, in \citet{borgs2007dynamics}, the authors develop t\^atonnement-style dynamics similar to those described in Algorithm~\ref{alg:dynamics}, which converge to an approximate pacing equilibrium in polynomial time. In combination with Theorem~\ref{thm:first-price-dynamics}, this provides evidence supporting the tractability of budget management for first-price auctions.

The uniqueness of pacing equilibrium and throttling equilibrium in first-price auctions is conducive to comparison, which we carry out for revenue. More specifically, in Theorem~\ref{thm:revenue-comparison}, we show that the revenue under the pacing equilibrium and the throttling equilibrium are always within a multiplicative factor of 2 of each other. Let REV(PE) and REV(TE) denote the revenue under the unique pacing equilibrium and the unique throttling equilibrium respectively. 

\begin{theorem}\label{thm:revenue-comparison}
    Given a throttling game $\left(n, m, (b_{ij}), (B_i)\right)$ for first-price auction, the revenue from the pacing equilibrium and the revenue from the throttling equilibrium are always within a factor of 2 of each other, i.e., REV(PE) $\leq 2 \times \text{REV(TE)}$ and REV(TE) $\leq 2 \times \text{REV(PE)}$.
\end{theorem}

In Appendix~\ref{appendix:TE-PE}, we give examples to demonstrate that REV(TE) can be arbitrarily close to $2 \times \text{REV(PE)}$, and REV(PE) can be arbitrarily close to $(4/3) \times \text{REV(TE)}$. In other words, for Theorem~\ref{thm:revenue-comparison}, the inequality REV(TE) $\leq 2 \times \text{REV(PE)}$ is tight and the inequality REV(PE) $\leq 2 \times \text{REV(TE)}$ is not too loose.

\paragraph{Comparing Pacing and Throttling for Second-Price Auctions}
This subsection is devoted to the comparison of pacing equilibria and throttling equilibria in second-price auctions. We begin by noting that, in stark contrast to first-price auctions, there could be infinitely many throttling equilibria for second-price auctions, as \Cref{ex: inf many spte} shows.

Next, we show that revenue and utility under truthful bidding are incomparable between throttling and pacing, i.e., neither method systematically dominates the other in either utility or revenue. In particular, the examples demonstrate instances where throttling yields higher utility for all buyers but lower revenue for the seller, and vice versa.
\begin{example}
    There are 2 goods and 3 buyers. The bids are given by $b_{11} = b_{12} = 102$, $b_{21} = 100$, $b_{22} = 0$, and $b_{31} = 50$, $b_{32} = 101$, and the budgets are given by $B_1 = 80$, $B_2 = 20$ and $B_3 = \infty$.

    First, we show that $\alpha_1 = 100/102$, $\alpha_2 = 1$ and $\alpha_3 = 1$ along with $x_{11} = 0.8$, $x_{21} = 0.2$ and $x_{32} = 1$ is a pacing equilibrium. To see this, note that 
    \begin{itemize}
        \item For good 1, buyers 1 and 2 both bid 100, buyer 3 bids 50. Therefore, both the highest bid and the second-highest bid are 100. As $x_{11} = 0.8$ and $x_{21} = 0.2$, we get that the payment made by buyer 1 and buyer 2 on good 1 is 80 and 20 respectively.
        \item For good 2, buyer 1 bids 100, buyer 2 bids 0 and buyer 3 bids 101. Therefore, buyer 3 wins the entire good and pays 100.
    \end{itemize}
    Hence, buyer 1 makes a total payment of 80, buyer 2 makes a total payment of 20 and buyer 3 makes a total payment of 100. These payments are budget feasible and satisfy the no-unnecessary pacing condition. Hence, $\alpha_1 = 100/102$, $\alpha_2 = 1$ and $\alpha_3 = 1$ with $x_{11} = 0.8$, $x_{21} = 0.2$ and $x_{32} = 1$ is a pacing equilibrium. Note that, if $v_{ij} = b_{ij}$, the utility of buyers 1, 2 and 3 under this pacing equilibrium is $x_{11}*2 = 1.6$, $x_{21}*0 = 0$ and $x_{32}*1 = 1$ respectively. The total revenue of the seller is $80 + 20 + 100 = 200$

    Next, consider a throttling equilibrium $(\theta_1, \theta_2, \theta_3)$ for this instance. Since buyer 3 has an infinite budget, the no-unnecessary throttling condition implies $\theta_3 = 1$. Now, if $\theta_1 = 1$, then buyer 1 would pay 101 on good 2, thereby violating her budget of $80$. Therefore, we must have $\theta_1 < 1$ and buyer 1 must exactly spend her budget of 80 in expectation, i.e., we must have $101 \cdot \theta_1 + 100 \cdot \theta_1 \theta_2 + 50 \cdot \theta_1 (1 - \theta_2) = 80$.  
    If we set $\theta_2 = 1$, then we get that $\theta_1 = 80/201$, which in turn implies that buyer 2's payment on good 1 is $50 \cdot (1 - \theta_1) \theta_2 > 20$, and contradicts the budget constraint of buyer 2. Therefore, we must have $\theta_2 < 1$, and consequently buyer 2's budget constraint must be tight: $50 \cdot (1 - \theta_1)\theta_2 = 20$.
    Solving the above equations for $\theta_1$ and $\theta_2$ yields $\theta_1 \approx 0.43$ and $\theta_2 \approx 0.7$. Hence, if $v_{ij} = b_{ij}$, the utility of buyer 1, 2 and 3 under this throttling equilibrium is at least $52 \cdot \theta_1 (1 - \theta_2) = 6.7$, $50 \cdot (1 -\theta_1) \theta_2 = 19.95$, and $101 \cdot (1 - \theta_1) = 57.57$. The total revenue of the seller is $80 + 20 = 100$. Thus, when compared to the pacing equilibrium, the throttling equilibrium results in higher utility for \emph{all} buyers, but lower revenue for the seller.
\end{example}

The above example showed that throttling can yield higher utility for all advertisers compared to pacing, i.e., all buyers can be better off if they use throttling instead of pacing. For the converse, observe that Example~\ref{ex: inf many spte} already exhibits a pacing equilibrium that yields higher utility for all buyers compared to a throttling equilibria. In particular, under truthful bidding ($v_{ij} = b_{ij}$) the symmetric pacing equilibrium with $\alpha_1 = \alpha_2 = 0.5$ yields a utility of $1.5$ for each buyer, whereas the symmetric throttling equilibrium with $\theta_1 = \theta_2 = \sqrt{0.5}$ gives a strictly lower utility of $\approx 1.12$ to both buyers. To complete the picture, the next example shows that the seller can receive higher revenue under throttling compared to pacing.

\begin{example}
    There are 2 goods and 3 buyers. The bids are given by $b_{11} = b_{12} = 101$, $b_{21} = b_{22} = 100$ and $b_{31} = 102$, $b_{32} = 0$, and the budgets are given by $B_1 = 25$, $B_2 = 25$ and $B_3 = \infty$.

    First, observe that $\alpha_1 = 50/101$, $\alpha_2 = 50/100$ and $\alpha_3 = 1$ along with $x_{31} = 1$ and $x_{12} = x_{22} = 0.5$ is a pacing equilibrium. This is because under this profile of pacing multipliers and tie-breaking rules, buyer 3 wins good 1 entirely at a price of 50, and good 2 is shared evenly between buyer 1 and 2 at a price of $50$. Therefore, all buyers satisfy their budget constraints, buyer 3 has $\alpha_3 =1$, and buyers 1 and 2 spend exactly their budget. In this pacing equilibrium, the total revenue of the seller is $100$.

    On the other hand, observe that $\theta_1 = 0.5$ and $\theta_2 = \theta_3 = 1$ is a throttling equilibrium. To check this, note that the total payment of buyer 1, 2 and 3 under this profile of throttling parameters is given by $100\cdot \theta_1\theta_2 = 50$, $0$ and $101\cdot \theta_1 \theta_3 + 100 \cdot (1- \theta_1)\theta_2 \theta_3 \geq 100$ respectively. Therefore, each buyer satisfies their budget constraints and either has $\theta_i = 1$ (buyers 2 and 3) or exactly spends her budget (buyer 1). In this throttling equilibrium, the total revenue of the seller is at least $150$, which is higher than the pacing equilibrium.
\end{example}

Finally, we provide a brief summary of the computational/structural similarities between pacing and throttling in second-price auctions. \cite{conitzer2018multiplicative} demonstrate that multiplicity (although only finitely-many different equilibria) also shows up for pacing equilibria in second-price auctions, which in combination with the multiplicity of throttling equilibria bodes unfavorably for potential comparisons of revenue. The similarities do not end with multiplicity: the problems of computing an approximate pacing equilibrium and computing an approximate throttling equilibrium are both PPAD-complete \citep{chen2021complexity}. As a consequence, we get that, unlike first-price auctions, no dynamics can converge to an approximate equilibrium in polynomial time for second-price auctions under either budget-management approach (assuming PPAD-hard problems cannot be solved in polynomial time). Furthermore, finding the revenue maximizing throttling equilibrium and finding the revenue maximizing pacing equilibrium are both NP-hard problems~\citep{conitzer2018multiplicative}. However, unlike throttling equilibria, a rational pacing equilibrium always exists \citep{chen2021complexity} and can be computed using a Mixed-Integer program~\citep{conitzer2018multiplicative}.

%% file: poa.tex
\section{Price of Anarchy}\label{sec:poa}

In this section, we study the efficiency of throttling equilibria in first-price and second-price auctions. We will use liquid welfare \citep{dobzinski2014efficiency} to measure efficiency. It is an alternative to social welfare which is more suitable for settings with budget constraints, and it reduces to social welfare when budgets are infinite.

\begin{definition}\label{definition:liquid-welfare}
    For an allocation $y = \{y_{ij}\}$, where $y_{ij} \in [0,1]$ denotes the probability of allocating good $j$ to buyer $i$, its liquid welfare $\lw(x)$ is defined as
\begin{align*}
    \lw(y) = \sum_{i=1}^n \min\left\{ \sum_{j=1}^m b_{ij} y_{ij}, B_i \right\}
\end{align*}
\end{definition}
\begin{remark}
    Liquid welfare is traditionally defined as the amount of revenue that can be extracted from budget-constrained buyers with full knowledge of their values. If buyer $i$ was assumed to have value $v_{ij}$ for good $j$, this is given by $\sum_{i=1}^n \min\left\{ \sum_{j=1}^m v_{ij} y_{ij}, B_i \right\}$. However, since our model does not assume a valuation structure, we define $\lw(y)$ to capture the amount of revenue that can be extracted from budget-constrained buyers with full knowledge of their bids if no buyer could be charged more than her bid for any good. It reverts to the traditional definition when $b_{ij} = v_{ij}$. \revision{A similar perspective---analyzing properties of a mechanism as functions of reported utilities rather than true utilities---also appears in the stable matching~\citep{roth1982economics, roth1999redesign} and core-selecting auction literature~\citep{day2008core, niazadeh2022fast}, which study complex environments where fully strategic misreporting is difficult to model or optimize.}
\end{remark}

Let $y(\theta)$ denote the allocation when the buyers use the throttling parameters $\theta = (\theta_1, \dots, \theta_n)$, and let $\Theta^*$ be the set of all throttling equilibria. Price of Anarchy \citep{koutsoupias2009worst}, which we define next, is the ratio of the worst-case liquid welfare of throttling equilibria to the best-possible liquid welfare that can be attained by any allocation. It measures the worst-case loss in efficiency incurred due the strategic behavior of agents when compared to the optimal outcome that could be achieved by a central planner.

\begin{definition}
    The Price of Anarchy (PoA) of throttling equilibria for liquid welfare is given by
    \begin{align*}
        \text{PoA} = \frac{\sup_{y \in (\Delta^n)^m} LW(y)}{\inf_{\theta \in \Theta^*} \lw(y(\theta))}
    \end{align*}
\end{definition}

The following theorem characterizes the Price of Anarchy for both first-price and second-price auctions. 

\begin{theorem}\label{thm:poa}
    For both first-price and second-price auctions, we have $\text{PoA} \leq 2$. Moreover, this bound is tight, i.e., there are instances for both first-price and second-price auctions where $\text{PoA}$ is arbitrarily close to $2$.
\end{theorem}

To establish the bound $\text{PoA} \leq 2$, we leverage the no-unnecessary-throttling condition of throttling equilibria; the proof of this part is inspired by the Price of Anarchy result of \citet{balseiro2021contextual} for pacing equilibria. Showing tightness requires the construction of a family of instances which have $\text{PoA} = 2 - \epsilon$ for arbitrarily small $\epsilon > 0$. To do so, we develop a novel and intricate construction which may be of independent interest. The full proof of Theorem~\ref{thm:poa} can be found in Appendix~\ref{appendix:poa}.

\section{Numerical Experiments}\label{sec:numerics}

In this section, we use semi-synthetic data to numerically analyze some properties of throttling on realistic instances. For our experiments, we use the dataset of \citet{liao2023statistical}, who employed logistic regression to create a model of click-through-rates for the iPinYou dataset~\citep{liao2014ipinyou}. In particular, we use their logistic regression model to estimate the click-through-rates (CTR) for 5 advertisers (buyers) and 1.8 million user impressions (items). For tractability, we used k-means clustering on the 1.8 million impressions to create $m \in \{10,50,100\}$ representative impressions. Given the variability in CTR,  we set the bid-per-click to be 1 for all advertiser and simply set their bids equal to their CTRs. Given the computational difficulties associated with second-price auctions, namely the absence of an efficient procedure for computing/evaluating all throttling equilibria, we focus our experiments on first-price auctions. We used Algorithm~\ref{alg:dynamics} initialized at $\theta_i = 0.1$ to compute throttling equilibria for all our experiments. We compute pacing equilibria using the convex program from \citet{conitzer2019pacing}.

\textbf{Revenue Comparison in First-Price Auctions.} We computed both the pacing and throttling equilibria for the three semi-synthetic instances with $m \in \{10,50,100\}$ items. Furthermore, for each instance, we assigned a high budget to two advertisers and low binding budgets to other advertisers; this created 10 instances (all possible combinations of high/low budgets) each for all $m \in \{10,20,30\}$. We found that the revenue under the throttling equilibrium was slightly higher for all instance sizes: $4.1\% \pm 0.6$ higher for $m = 10$ items, $2.5\% \pm 0.9$ higher for $m = 50$ items, and $2.2\% \pm 0.7$ for $m=100$ items, where the term after the $\pm$ denotes the standard deviation over all possible budget assignments.

\begin{figure}[t!]
	\begin{minipage}{0.33\textwidth}
		\includegraphics[width = \linewidth]{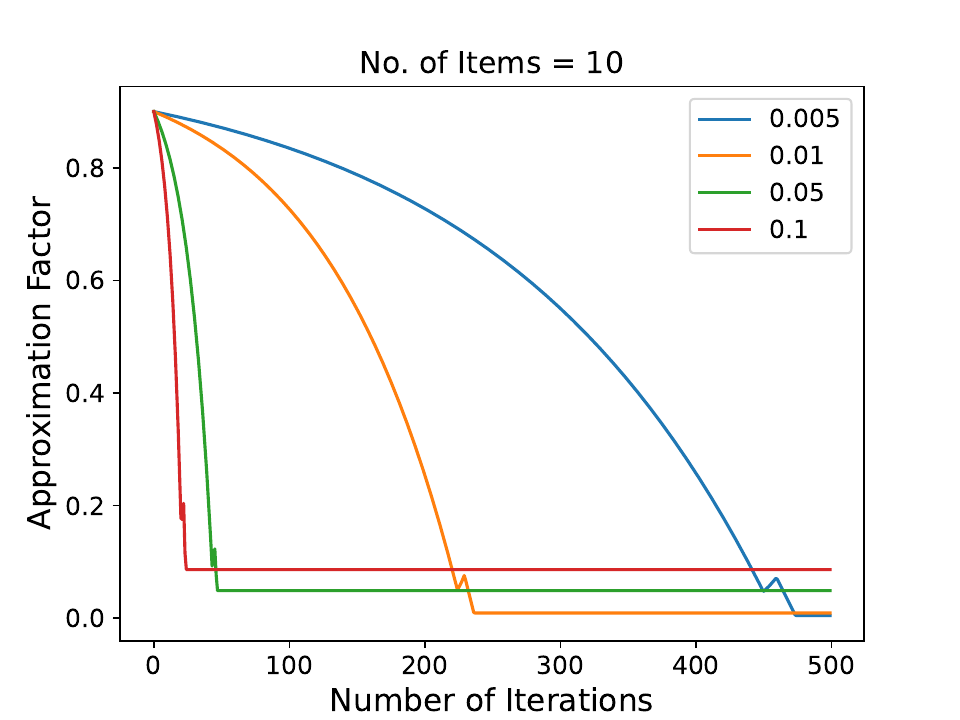}
	\end{minipage}%
	\begin{minipage}{0.33\textwidth}
		\includegraphics[width = \linewidth]{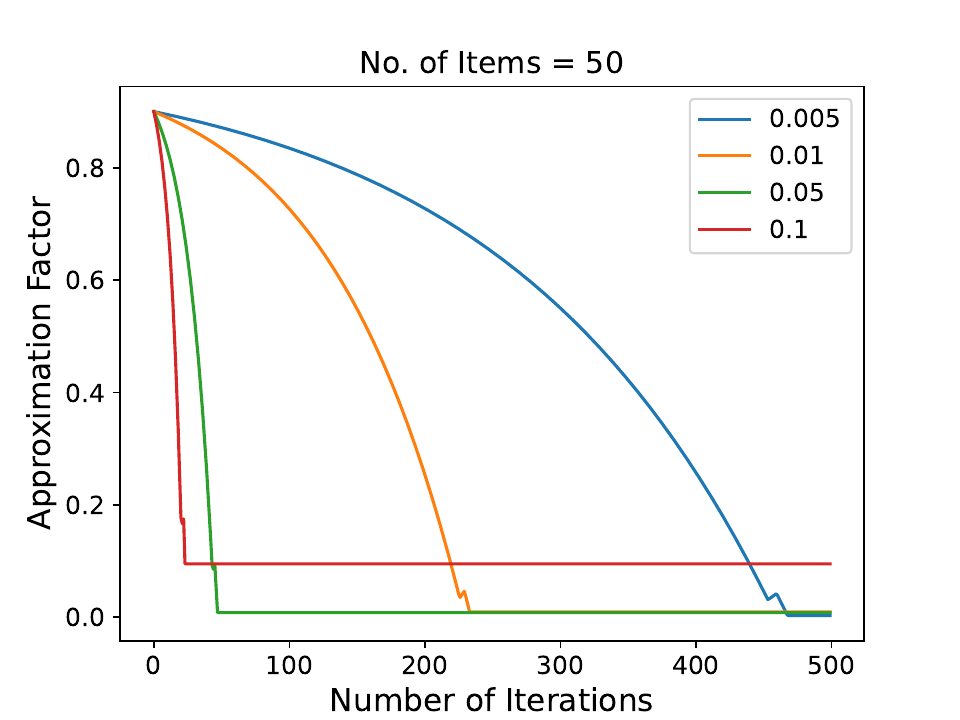}
	\end{minipage}%
	\begin{minipage}{0.33\textwidth}
		\includegraphics[width = \linewidth]{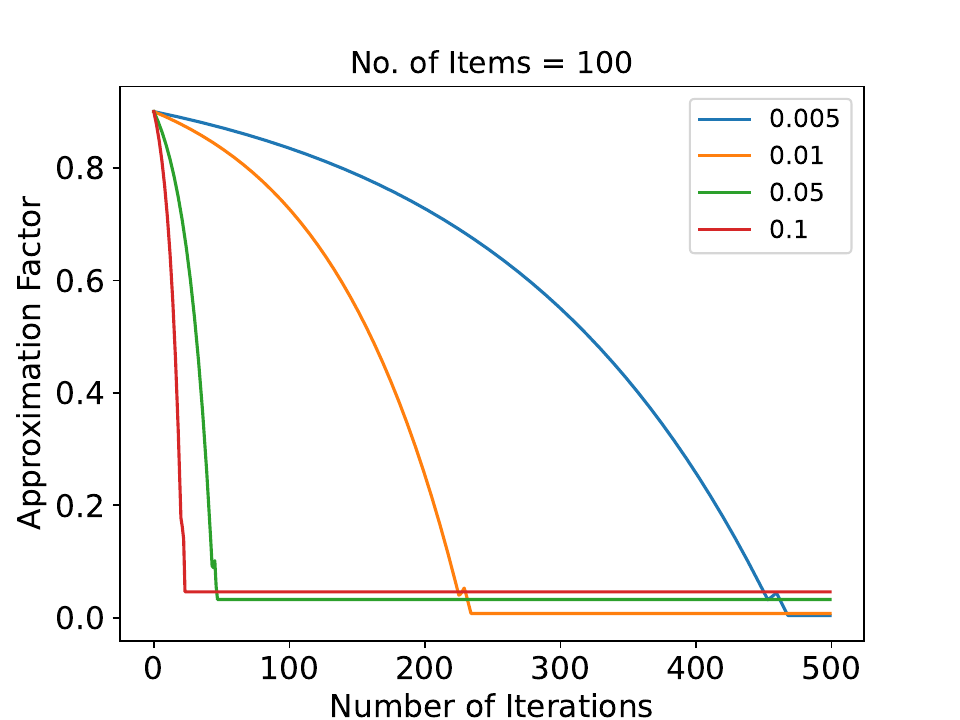}
	\end{minipage}
    \caption{Convergence rate of Algorithm~\ref{alg:dynamics} for different step sizes $\delta \in \{0.005, 0.01, 0.05, 0.1\}$.}
    \label{fig:convergence_fpa}
\end{figure}

\textbf{Convergence Rate of Algorithm~\ref{alg:dynamics}.} Next we evaluated the convergence properties of Algorithm~\ref{alg:dynamics} for different step sizes $\delta = \{0.005, 0.01, 0.05, 0.1\}$. Once again, the budget was set to be large for 2 random buyers and low for the others. Figure~\ref{fig:convergence_fpa} depicts the results; it plots the approximation factor $\max_{i \in [n]} \min\left\{ \sum_{j} p_{ij}(\theta), 1 - \theta_i  \right\}$ as a function of the number of iterations of Algorithm~\ref{alg:dynamics}. As expected, smaller step sizes lead to faster convergence, but to an equilibrium with a worse approximation factor. Interestingly, like our theoretical rate, the number of items does not impact the numerical rate of convergence, and the approximation factor monotonically decreases for the most part.

\textbf{Representativeness of Items Won.} We stated earlier that one of the practical advantages of throttling is its ability to provide advertisers with a more representative sample of users. We numerical evaluate this claim on our semi-synthetic data. In particular, under the same setup as the revenue comparison described above, we computed the number of items won by advertisers with probability 1\% or higher. We found that for $m=10$ items, advertisers won $5.2 \pm 0.9$ under throttling and $2.5 \pm 0.4$ under pacing on average; for $m=50$ items, advertisers won $29 \pm 3$ under throttling and $10 \pm 2$  under pacing on average; and for $m=100$ items, advertisers won $58 \pm 3$ under throttling and $20 \pm 3$ under pacing on average. Here, the term after the $\pm$ denotes the standard deviation over the 5 buyers. Thus, in all instances, throttling led to advertisers winning at least twice as many items with a 1\% probability or higher, thereby yielding a more representative sample.


%% file: conclusion.tex
\section{Conclusion and Future Work}\label{sec:conclusion}

We defined the notion of a throttling equilibrium and studied its properties for both first-price and second-price auctions. Through our analysis of computational and structural properties, we found that throttling equilibria in first-price auctions satisfy the desirable properties of uniqueness and polynomial-time computability. 
In contrast, we showed that for second-price auctions, equilibrium multiplicity may occur, and computing a throttling equilibrium is PPAD hard.
This disparity between the two auction formats is reinforced when we compare throttling and pacing: our results show that the properties of throttling equilibrium across the two formats have a striking similarity to the properties of first-price versus second-price pacing equilibrium. Finally, we also showed that the Price of Anarchy of throttling equilibria for liquid welfare is bounded above by 2 for both first-price and second-price auctions, and that this bound is tight for both auction formats. Altogether, this provides a comprehensive analysis of the equilibria which arise from the use of throttling as a method of budget management.

There are many interesting directions for future work, such as what happens when a combination of pacing and throttling-based buyers exist in the market, whether the combination of throttling and pacing behaves well for second-price auctions, whether second-price throttling equilibria can be computed efficiently under some natural assumptions on the bids, what is the smallest $\delta$ for which a $\delta$-approximate throttling equilibrium can be computed efficiently, \revision{what is the complexity of computing exact throttling equilibria, how do item-dependent throttling parameters impact equilibrium outcomes\footnote{There is an entire spectrum of participation-probability heterogeneity between the single-parameter setting analyzed here and the complete freedom for subset selection considered in \citet{charles2013budget}.}, what is the Price of Anarchy in the throttling metagame where buyers also strategize over their bids/budgets,} and whether the tractability of budget management in first-price auctions holds more generally beyond throttling and pacing.

%% file: appendix-irrationality.tex
\section{Appendix: Examples of Irrationality and the Nature of Approximation}\label{appendix:irrationality}

\paragraph{First-Price Auctions:} First, we give an example for which the unique \emph{first-price} throttling equilibrium is irrational.

\begin{example}\label{example:first-price-irrational}
	Define a throttling game as follows:
	 There are 2 goods and 2 buyers, i.e., $m = 2$ and $n = 2$;
	 $b_{11} = b_{12} =  2$ and $b_{21} = 1, b_{22} = 3$;
    $B_1 = 2$ and $B_2 = 1$.
	Suppose, in equilibrium, the buyers use the throttling parameters $\theta_1$ and $\theta_2$. Then the payment of buyer 1 and buyer 2 are given by $2\theta_1 + 2(1 - \theta_2)\theta_1$ and $3\theta_2 + (1 -  \theta_1) \theta_2$ respectively. Therefore, for this game, in any throttling equilibrium, we have $0 < \theta_1, \theta_2 < 1$, which implies
	  $  2\theta_1 + 2(1 - \theta_2)\theta_1 = 2$ and $3\theta_2 + (1 -  \theta_1) \theta_2 = 1
	$.
	Substituting $\theta_1 = 1/(2 - \theta_2)$ from the first equation into the second   yields
	\begin{align*}
 	   3\theta_2 + \theta_2 \cdot \frac{1 - \theta_2}{2 - \theta_2} = 1 
	\end{align*}
	which implies $4 \theta_2^2 - 7 \theta_2 + 1 = 0$. As $\theta_2 < 1$, Solving the quadratic gives $\theta_2 = (7 - \sqrt{33})/8$.
\end{example}

\revision{
\begin{example}[A $\delta$-approximate throttling equilibrium can be far in $\ell_\infty$]\label{ex:weak-approx-fpa}
Fix any $\delta\in(0,1/2)$ and consider a \emph{first-price} throttling game with $n=2$ buyers and $m=1$ good.
Let the bids be $b_{11}=2,\ b_{21}=1$, and the budgets be $B_1 \;=\; 4 \cdot (1-\delta),\
B_2 \;=\; \frac{\delta^2}{2(1-\delta)}$. For $\theta=(\theta_1,\theta_2)\in[0,1]^2$, the expected payments on the (single) good are
\[
p(\theta)_{11} \;=\; 2\theta_1,\qquad
p(\theta)_{21} \;=\; \theta_2(1-\theta_1).
\]

Since $B_1>2$, buyer~1 spends less than the budget $B_1$ for every $\theta_1\in[0,1]$:
\(
2\theta_1 \le 2 < B_1.
\)
Thus, by the ``no unnecessary throttling'' condition, any throttling equilibrium must satisfy $\theta_1=1$.
When $\theta_1=1$, buyer~2's expected payment is $0<B_2$, so any throttling equilibrium must also satisfy $\theta_2=1$.
Therefore, $\theta^* \;=\; (1,1)$ is the unique throttling equilibrium.

However, there exists a $\delta$-approximate throttling equilibrium, namely $\hat\theta \;\coloneqq\; (1-\delta,\ \delta/2)$, which is far away from $\theta^*$ in $\ell_\infty$ distance. To see $\hat \theta$ is indeed a $\delta$-approximate throttling equilibrium, note that
Buyer~1's expected payment is $2(1-\delta)\le B_1$, and buyer~2's expected payment is
\[
\frac{\delta}{2}(1-(1-\delta)) \;=\; \frac{\delta}{2}\cdot \delta \;=\; \frac{\delta^2}{2}
\;\le\; \frac{\delta^2}{2(1-\delta)} \;=\; B_2,
\]
so $\hat\theta$ is budget-feasible.
Moreover, buyer~1 satisfies $\hat\theta_1=1-\delta$, and buyer~2 satisfies
\[
p(\hat\theta)_{21} \;=\; \frac{\delta^2}{2} \;=\; (1-\delta)B_2,
\]
so $\hat\theta$ satisfies the ``not too much unnecessary throttling'' condition in the definition of a
$\delta$-approximate throttling equilibrium.
Hence, $\hat\theta$ is a $\delta$-approximate throttling equilibrium. However,
\[
\|\hat\theta-\theta^*\|_\infty
\;=\; \max\left\{\;|1-(1-\delta)|,\;|1-\delta/2|\;\right\}
\;=\; 1-\delta/2,
\]
which can be arbitrarily close to $1$ for small $\delta$.
\end{example}

\begin{remark}\label{rem:alg-can-hit-weak-approx}
For the instance in Example~\ref{ex:weak-approx-fpa}, Algorithm~\ref{alg:dynamics} (with parameter $\delta$) initializes
\[
\theta_1=\min\!\left\{\frac{B_1}{2\sum_j b_{1j}},1\right\}=1-\delta,
\]
and buyer~1 is never updated (since the update condition requires $\theta_i<1-\delta$).
Buyer~2 is increased multiplicatively until its expected payment reaches $(1-\delta)B_2$. Thus, the algorithm terminates with
\[
\theta_2 \in \left[\frac{\delta}{2},\ \frac{\delta}{2(1-\delta)}\right]\,.
\]
Therefore, this example shows that Algorithm~\ref{alg:dynamics} can output a $\delta$-approximate throttling equilibrium at $\ell_\infty$-distance at least
$1-\delta/(2(1-\delta))$ from the unique exact throttling equilibrium $\theta^*=(1,1)$.
\end{remark}
}

\paragraph{Second-Price Auctions:} Next, we give an example for which all \emph{second-price} throttling equilibria are irrational.

\begin{example}\label{example:irrational_eq}
	Define a throttling game as follows:
	\begin{itemize}
    	\item There are 4 goods and 3 buyers, i.e., $m = 4$ and $n = 3$
    	\item $b_{11} = b_{12} = 2$, $b_{14} = 1$, $b_{23} = b_{24} = 4$, $b_{22} = 1$, $b_{31} = 1$ and $b_{33} = 2$
    	\item $B_1 = B_2 = 1$ and $B_3 = \infty$
	\end{itemize}
	For this game, in any throttling equilibrium, we have $0 < \theta_1, \theta_2 < 1$ and $\theta_3 = 1$. Hence, if $\theta$ is a throttling equilibrium, then it satisfies
	$	    \theta_1 + \theta_1 \theta_2 = 1$ and $2\theta_2 + \theta_2 \theta_1 = 1 $.
	Substituting $\theta_1 = 1/(1 + \theta_2)$ from the first equation into the second equation yields
	\begin{align*}
 	   2\theta_2 + \theta_2 \cdot \frac{1}{1 + \theta_2} = 1 
	\end{align*}
	which further implies $2 \theta_2^2 + 2 \theta_2 - 1 = 0$. As $\theta_2 > 0$, solving the quadratic gives $\theta_2 = (\sqrt{3} - 1)/2$.
\end{example}

\revision{
\begin{example}[A $\delta$-approximate equilibrium can be far in $\ell_\infty$]\label{ex:weak-approx-spa}
Fix any $\delta\in(0,1/2)$ and consider a \emph{second-price} throttling game with $n=3$ buyers and $m=1$ good.
Let bids be $b_{11}=3,\ b_{21}=2,\ b_{31}=1$,
and budgets be $B_1=3,\ B_3=1,\ B_2=\frac{\delta^2}{2(1-\delta)}$.
For $\theta=(\theta_1,\theta_2,\theta_3)\in[0,1]^3$, the expected payments under second-price are
\begin{align*}
p(\theta)_{11}
= \theta_1\Big(2\theta_2 + 1\cdot(1-\theta_2)\theta_3\Big),\quad p(\theta)_{21}
= (1-\theta_1)\theta_2\theta_3\cdot 1,\quad p(\theta)_{31}=0.
\end{align*}

Since $p(\theta)_{11}\le 2\theta_1\le 2 < B_1$ for all $\theta$, buyer~$1$ always spends less than her budget. Hence, any throttling equilibrium must satisfy $\theta_1=1$.
Also $p(\theta)_{31}=0<B_3$, so any throttling equilibrium must satisfy $\theta_3=1$.
With $\theta_1=1$, buyer~$2$ never wins and thus $p(\theta)_{21}=0<B_2$, implying $\theta_2=1$.
Therefore, $\theta^*=(1,1,1)$ is the unique throttling equilibrium.

However, there exists an approximate throttling equilibrium, namely $\hat\theta \coloneqq (1-\delta,\ \delta/2,\ 1)$, which is far from $\theta^*$. To see this, note that
\[
p(\hat\theta)_{21}=(1-(1-\delta))\cdot \frac{\delta}{2}\cdot 1
= \frac{\delta^2}{2}
= (1-\delta)B_2,
\]
so buyer~$2$ is budget-feasible and (despite $\hat\theta_2<1-\delta$) satisfies the
``not too much unnecessary throttling'' condition of a $\delta$-approximate throttling equilibrium.
Buyers~$1$ and $3$ have $\hat\theta_1=1-\delta$ and $\hat\theta_3=1$, so they satisfy the same condition trivially.
Budget-feasibility holds for buyer~$1$ as well since
\[
p(\hat\theta)_{11}=(1-\delta)\Big(2\cdot \frac{\delta}{2} + (1-\frac{\delta}{2})\cdot 1\Big)
=(1-\delta)\Big(1+\frac{\delta}{2}\Big) \le 3 = B_1,
\]
and for buyer~$3$ since $p(\hat\theta)_{31}=0\le B_3$.
Thus $\hat\theta$ is a $\delta$-approximate throttling equilibrium. Finally,
\[
\|\hat\theta-\theta^*\|_\infty
=\max\left\{|1-(1-\delta)|,\ |1-\delta/2|,\ |1-1|\right\}
=1-\delta/2,
\]
which can be arbitrarily close to $1$ for small $\delta$.
\end{example}
}

%% file: appendix-PPAD-membership.tex
\section{Appendix: Missing Proofs}

\subsection{Proof of Lemma~\ref{lemma:hardness}}\label{appendix:hardness}

Before proceeding onto the proof of Lemma~\ref{lemma:hardness}, we establish some key lemmas. Since we are considering a particular $\theta$, we will suppress the dependence on $\theta$ of the payment made by buyer $B$ on good $G$ and simply denote it by $p(B,G)$. The next lemma notes the payment terms of buyers in $\I_\G$.

\begin{lemma}\label{lemma:buyer_payments}
    For all $i \in [n]$, we have
    \begin{flushleft}\begin{enumerate}
        \item $p(T(i), G(i,j)) = \left(1 - \theta_{S(j)}\right)\theta_{T(i)}\theta_{S(i)} 4$, for all $j \in N_i$; and $p(T(i), R(i)) = 0$.
        \item $p(S(i), R(i)) = \theta_{S(i)} \theta_{T(i)} M |O_i|$;
        $p(S(i), G(i,j)) = 0$, for all $j \in N_i$; and
        $$p(S(i), G(j,i)) = \theta_{S(i)} \left[\theta_{T(j)} 5\ +\ (1 - \theta_{T(j)}) \theta_{S(j)} 4 \right],\quad\text{for all $j \in O(i)$.}$$
    \end{enumerate}\end{flushleft}
\end{lemma}

In the next lemma, we bound the total payment made by strategy buyer $S(i)$ on the neighbor goods and provide lower bounds on the throttling parameters:

\begin{lemma} \label{lemma:neighbour_payment_bound}
    For all $i \in [n]$, we have $\theta_{S(i)} \geq (1 - 2 \delta)/2$, $\theta_{T(i)} \geq 1/32 $, and
    the total payment of $S(i)$ on the neighbor goods satisfies:
            \begin{align*}
                |O_i| \theta_{S(i)} \leq \sum_{j \in N_i} p(S(i), G(i,j)) + \sum_{j \in O_i} p(S(i), G(j,i)) \leq 5 |O_i| \theta_{S(i)}
            \end{align*}
    %
\end{lemma}
\begin{proof}
    For $i \in [n]$, using Lemma~\ref{lemma:buyer_payments}, we get
    \begin{align*}
        \sum_{j \in N_i} p(S(i), G(i,j)) + \sum_{j \in O_i} p(S(i), G(j,i)) =  \sum_{j \in O_i} \theta_{S(i)} \left[\theta_{T(j)} 5\ +\ (1 - \theta_{T(j)}) \theta_{S(j)} 4 \right] \leq 5 |O_i| \theta_{S(i)}.
    \end{align*}
    Suppose $\theta_{S(i)} < (1 - 2\delta)/2$ for some $i \in [n]$. Then, the total payment made by $S(i)$ is at most 
    \begin{align*}
        \theta_{S(i)} \theta_{T(i)} M |O_i| + 5|O_i|\theta_{S(i)}
        < \frac{(1 - 2 \delta)M|O_i|}{2} + 5|O_i|
        < (1 - \delta) \cdot \frac{M|O_i|}{2}
    \end{align*}
    using $M>10/\delta$ by our choice of $M$.
    This contradicts the definition of $\delta$-approximate throttling equilibrium because $S(i)$ has a budget of $M|O_i|/2$. Therefore, we have $\theta_{S(i)}\ge (1-2\delta)/2$ and in particular, $\theta_{S(i)} \ge 1/4$ for all $i \in [n]$ using $\delta\le 1/4$. Hence,
    \begin{align*}
        \sum_{j \in N_i} p(S(i), G(i,j)) + \sum_{j \in O_i} p(S(i), G(j,i)) &=  \sum_{j \in O_i} \theta_{S(i)} \left[\theta_{T(j)} 5\ +\ (1 - \theta_{T(j)}) \theta_{S(j)} 4 \right]\\
        &> \sum_{j \in O_i} \theta_{S(i)} \left[\theta_{T(j)} \theta_{S(j)} 4\ +\ (1 - \theta_{T(j)}) \theta_{S(j)} 4 \right]\\
        &= 4\theta_{S(i)} \cdot \sum_{j \in O_i}  \theta_{S(j)} 
         \geq 
          |O_i|\theta_{S(i)}.
    \end{align*}
    
    Suppose, $\theta_{T(i)} < 1/32$. By Lemma~\ref{lemma:buyer_payments} the total payment of threshold buyer $T(i)$ is at most $$\sum_{j\in N_i} 4\theta_{T(i)}\le 12 \theta_{T(i)}<\frac{3}{8}\le \frac{1-\delta}{2}$$ using $|N_i|\le 3$ and $\delta\ge 1/4$. 
    Hence, we have obtained a contradiction with the definition of $\delta$-approximate throttling equilibria. Therefore, $\theta_{T(i)} \geq 1/32$.
\end{proof}

We are now ready to prove Lemma~\ref{lemma:hardness}.

\begin{proof}[Proof of Lemma~\ref{lemma:hardness}]
    Fix an $i \in [n]$. First consider the case when $\sum_{j \in N_i} x_j >0.5 + \epsilon$. Assume for a contradiction that $x_i > \epsilon$. Then, $\theta(S(i)) < 1 - (\epsilon/2) \leq 1 -\delta$ using $\delta\le \epsilon/2$. By the definition of $\delta$-approximate throttling equilibrium, the total payment of threshold buyer $S(i)$ is at least $(1 -\delta)M|O_i|/2$. Combining this observation with Lemma~\ref{lemma:buyer_payments} and Lemma~\ref{lemma:neighbour_payment_bound}, we get
    \begin{align*}
        \theta_{S(i)} \left[\theta_{T(i)} M |O_i| + 5|O_i|\right] \geq (1 -\delta) \cdot \frac{M|O_i|}{2}
    \end{align*}
and thus (using $\theta_{T(i)}\ge 1/32$ from Lemma \ref{lemma:neighbour_payment_bound}
  and our choice of $M=160/\delta$),
$$  \theta_{S(i)} \geq \frac{1-\delta}{2 \theta_{T(i)} + (10/M)} \geq \frac{1}{2 \theta_{T(i)}} \cdot \frac{1 -\delta}{1+\delta} \ \ \implies\ \  \theta_{T(i)}\theta_{S(i)} \geq \frac{1}{2 (1 + 2 \epsilon)}
   $$ using $\delta\le \epsilon/2$ and $\epsilon< 1$.
    Moreover, note that $\sum_{j \in N_i} x_j > 0.5 + \epsilon$ implies $$\sum_{j \in N_i} 2\left(1 - \theta_{S(j)}\right) > (1+2\epsilon)/2$$ Combining the above statements allows us to bound the total payment of buyer $T(i)$:
    \begin{align*}
        \sum_{j \in N_i} \left(1 - \theta_{S(j)}\right)\theta_{T(i)}\theta_{S(i)} 4 \geq \frac{4}{2 (1 + 2 \epsilon)} \cdot \sum_{j \in N_i} \left(1 - \theta_{S(j)}\right) > \frac{1}{2}.
    \end{align*}
    This yields a contradiction because $T(i)$ has budget  $1/2$. Hence $x_i \leq \epsilon$ when $\sum_{j \in N_i} x_j > 0.5 + \epsilon$.
    Next consider the case of $\sum_{j \in N_i} x_j < 0.5 - \epsilon$. The budget constraint of $S(i)$ and Lemma~\ref{lemma:neighbour_payment_bound} yield 
    \begin{align*}
        \theta_{S(i)} \left[\theta_{T(i)} M |O_i| + |O_i|\right] \leq \frac{M|O_i|}{2}
        \end{align*}
which implies that 
\begin{align}\label{heheeq2}\theta_{S(i)} \left[\theta_{T(i)} M |O_i| + \theta_{T(i)}|O_i|\right] \leq \frac{M|O_i|}{2}\ \ 
         \implies\ \  \theta_{T(i)}\theta_{S(i)} \leq \frac{1}{2(1 + (1/M))}< \frac{1}{2}. 
    \end{align}
By Lemma~\ref{lemma:neighbour_payment_bound}, we have $\theta_{S(j)} \geq (1 - 2 \delta)/2$ and thus, $2(1 - \theta_{S(j)}) \leq 1 + 2\delta$. Multiplying both sides by $(1 - 2\delta)$ yields $2(1 - \theta_{S(j)})(1 - 2\delta) \leq 1 - 4\delta^2 < 1$. In other words, we have
\begin{equation}\label{heheeq3}
2(1 - \theta_{S(j)})(1 - 2\delta) < \min\{2(1 - \theta_{S(j)}), 1\} = x_j
\end{equation}
for every $j\in [n]$. This together with $\sum_{j \in N_i} x_j < 0.5 - \epsilon$ implies $$(1 - 2\delta)\sum_{j \in N_i} 2\left(1 - \theta_{S(j)}\right) < (1-2\epsilon)/2.$$ Therefore, we get that the total payment of $T(i)$ satisfies the following bound
    \begin{align*}
        \sum_{j \in N_i} \left(1 - \theta_{S(j)}\right)\theta_{T(i)}\theta_{S(i)} 4 < \sum_{j \in N_i} 2\left(1 - \theta_{S(j)}\right) < \frac{1 - 2\epsilon}{2(1 - 2\delta)} \leq (1 - \delta) \cdot \frac{1}{2} 
    \end{align*}
    using $\delta\le \epsilon/2$.
    As a consequence of the definition of $\delta$-approximate throttling equilibria, we have that $\theta_{T(i)} \geq 1 -\delta$. Finally, using (\ref{heheeq2}) and (\ref{heheeq3}), we have
    \begin{align*}
        x_i > 2(1 - \theta_{S(i)}) (1 - 2 \delta) \geq 2\left(1 - \frac{1}{2\theta_{T(i)}}\right) (1 - 2 \delta) \geq \frac{(1 - 2\delta)^2}{1 - \delta} > \frac{1 - 4 \delta}{1 - \delta} \geq 1 - \epsilon,
    \end{align*}
    where the last inequality follows from $\delta \leq \epsilon/(3 + \epsilon)$.
\end{proof}



\subsection{Proof of Theorem~\ref{thm:second-price-PPAD-membership}}\label{appendix:PPAD-membership}

	Consider a throttling game $\left(n, m, (b_{ij}), (B_i)\right)$ and an approximation parameter $\delta \in (0,1/2)$. Define $f: [0,1]^n \to [0,1]^n$ as
		\begin{align*}
    		f_i(\theta) = \min\left\{ \frac{(1 - \delta/2)B_i}{\sum_j p(1, \theta_{-i})_{ij}}, 1 \right\} = \min\left\{ \frac{(1 - \delta/2)B_i}{\max\{\sum_j p(1, \theta_{-i})_{ij}, B_i/2\} }, 1 \right\} \quad \forall \theta \in [0,1]^n 
		\end{align*}
	
	First, we prove that $f$ is $L$-Lipschitz continuous with Lipschitz constant $L = 2mn \overline{B} \underline{B}^{-2} \overline{b}$, where $\overline{b} = \max_{i,j} b_{ij}$, $\overline{B} = \max_{i} B_i$. To achieve this, we will repeatedly use the following facts about Lipschitz functions. For Lipschitz continuous functions $f$ and $g$ with Lipschitz constants $L_1$ and $L_2$ respectively,
	\begin{itemize}
		\item $f + g$ is $L_1 + L_2$-Lipschitz continuous
		\item If $f$ and $g$ are bounded above by $M$, then $f \cdot g$ is $M(L_1 + L_2)$-Lipschitz continuous
		\item If $f$ is bounded below by $c$, then $1/f$ is $L_1/c^2$-Lipschitz continuous
		\item For a constant $C$, $\max\{f, C\}$ and $\min\{f,C\}$ are both $L_1$-Lipschitz continuous
	\end{itemize}
	Observe that
	\begin{align*}
	    p(1, \theta_{-i})_{ij} =\sum_{\ell: b_{\ell j} < b_{ij}}  b_{\ell j} \theta_\ell  \prod_{k \neq i: b_{kj} > b_{\ell j}} (1 - \theta_k)
	\end{align*}
	Therefore, for all $i \in [n]$, $\theta \mapsto p(1, \theta_{-i})_{ij}$ is $(2n\overline{b})$-Lipschitz continuous, which further implies that $\theta \mapsto \sum_{j} p(1, \theta_{-i})_{ij}$ is $2mn\overline{b}$-Lipschitz continuous. Finally, due to the second equality in the definition of $f$, we get that $f$ is $(2mn \overline{B} \underline{B}^{-2} \overline{b})$-Lipschitz continuous.
	
	Since BROUWER is in PPAD~\citep{chen2009settling}, to complete the proof, it suffices to show that a $(\delta \underline{B}/4m \overline{b})$-approximate fixed point $\theta^*$ of $f$, i.e, $\theta^*$ such that $\|f(\theta^*) - \theta^*\|_\infty \leq \delta \underline{B}/4m \overline{b}$, is a $\delta$-approximate throttling equilibrium. First, note that $p(1, \theta_{-i})_{ij} \leq \overline{b}$ for all $i \in [n], j \in [m]$. Therefore, $f(\theta)_i \geq \underline{B}/2m \overline{b}$ for all $i \in [n]$. Hence, for $i \in [n]$, we have
	\begin{align*}
		\biggr\lvert 1 - \frac{\theta_i^*}{f_i(\theta^*)} \biggr\rvert \leq 	\frac{\delta \underline{B}}{f_i(\theta^*) \cdot 4m \overline{b}} \leq \frac{\delta}{2}
	\end{align*}
	As a consequence, we get $\theta_i^* \leq (1 + \delta/2) f_i(\theta^*)$ and $ \theta^* \geq (1- \delta/2) f_i(\theta^*)$. The first inequality implies  which in turn implies
	\begin{align*}
		\sum_j p(\theta^*)_{ij} = \theta^*_i \cdot \sum_j p(1, \theta^*)_{ij} \leq (1 + \delta/2)(1 - \delta/2)B_i \leq B_i
	\end{align*}
	 and the second one implies that if $\theta^*_i < 1 - \delta/2$, then
	 \begin{align*}
	 	\sum_j p(\theta^*)_{ij} = \theta^*_i \cdot \sum_j p(1, \theta^*)_{ij} \geq (1 - \delta/2)^2 B_i \geq (1 - \delta)B_i
	 \end{align*}
	 Hence, $\theta^*$ is a $\delta$-approximate throttling equilibrium, thereby completing the proof.\qed


%% file: appendix-NP-hard.tex
\subsection{Proof of Theorem~\ref{thm:revenue_np_hard}}\label{appendix:NP-hard}

Consider an instance of 3-SAT with variables $\{x_1, \dots, x_n\}$ and clauses $\{C_1, \dots, C_m\}$. Our goal is to define an instance $\mathcal{I}$ of REV (a throttling game $G$ and a  target revenue $R$) which always has the same solution (Yes or No) as the 3-SAT instance, and has a size of the order $\text{poly}(n,m)$. We do so next, starting with an informal description to build intuition. To better understand the core motivations behind the gadgets, we will restrict our attention to exact throttling equilibria ($\delta = 0$) in the informal discussion that follows. As we will see in the formal proof, the target revenue $R$ can be chosen carefully to ensure that only exact throttling equilibria can achieve the revenue $R$. 

\textbf{Reciprocal Gadget:} Fix $i \in [n]$. Corresponding to variable $x_i$, there are two goods $\mathbb{A}_i$ and $\mathbb{B}_i$, and two buyers $V_i^+$ and $V_i^-$ in the throttling game $G$. Each buyer bids 1 for one of the goods and bids 2 for the other, with both buyers bidding differently on each good. Furthermore, we set the budgets of both buyers to be 1/2, and ensure that they do not spend any non-zero amount on goods other than $\A_i$ and $\B_i$. In equilibrium, this forces the throttling parameter of $V_i^+$ (which we denote by $\theta_i^+$) to be half of the reciprocal of the throttling parameter of $V_i^-$ (which we denote by $\theta_i^-$) and vice-versa. As a consequence, both throttling parameters lie in the interval $[1/2,1]$.

\textbf{Binary Gadget:} For each variable $x_i$, there are two additional goods $\S_i$ and $\T_i$, which receive a bid of 1 from buyers $V_i^+$ and $V_i^-$ respectively. The throttling game $G$ also has one unbounded buyer $U$ who has an infinite budget, and bids 2 on both goods $\S_i$ and $\T_i$. By the definition of throttling equilibria (Definition~\ref{def:exact_throt_eq}), the throttling parameter of $U$ is always 1 in equilibrium. Therefore, buyer $U$ wins both $\S_i$ and $\T_i$ with probability one, and pays $\theta_i^+ + \theta_i^-  = \theta_i^+ + 1/2\theta_i^+$ for it. Finally, observe that $t \mapsto t + 1/2t$, when restricted to $t \in [1/2,1]$, is maximized at $t = 1$ or $t = 1/2$. Therefore, by appropriately choosing the target revenue $R$, we can ensure that revenue $R$ is only achieved by throttling equilibria in which exactly one of the following holds: $(\theta_i^+ = 1, \theta_i^- = 1/2)$ or $(\theta_i^+ = 1/2, \theta_i^- = 1)$. This allows us to interpret $\theta_i^+ = 1$ as setting $x_i = 1$ and $\theta_i^- = 1$ as setting $x_i = 0$.

\textbf{Clause Gadget:} For each clause $C_j$, there is a good $\C_j$. If $C_j$ contains a non-negated literal $x_i$, then buyer $V_i^+$ bids 1 on good $\C_j$, and if it contains a negated literal $\neg x_i$, then buyer $V_i^-$ bids 1 on good $\C_j$. Furthermore, the unbounded buyer $U$ bids 2 on good $\C_j$, thereby always winning it. Hence, the total payment on good $\C_j$ is 1 if some literal is satisfied (corresponding throttling parameter is 1), and is 1/2 if no literal is satisfied (corresponding throttling parameters are 1/2). The rest of the reduction boils down to choosing $R$ appropriately.

\vspace{1em}

\begin{proof}[Proof of Theorem~\ref{thm:revenue_np_hard}]
Guided by the informal intuition described above, we proceed with the formal definition of the instance $\mathcal{I}$, which involves specifying the throttling game $G$ and the target revenue $R$. The throttling game $G$ consists of the following goods:
\begin{itemize}
    \item \textbf{Reciprocal Gadget:} For each variable $x_i$, there are two goods $\mathbb{A}_i$ and $\mathbb{B}_i$.
    \item \textbf{Binary Gadget:} For each variable $x_i$, there are two binary goods $\S_i$ and $\T_i$.
    \item \textbf{Clause Gadget:} For each clause $C_j$, there is a good $\mathbb{C}_j$.
\end{itemize}
Moreover, $G$ has the following set of buyers:
\begin{itemize}
    \item Corresponding to each variable $x_i$, there are two buyers $V_i^+$ and $V_i^-$ with non-zero bids only for the following goods:
    \begin{itemize}
        \item $b(V_i^+, \A_i) = 2$ and $b(V_i^+, \B_i) = 1$
        \item $b(V_i^-, \A_i) = 1$ and $b(V_i^-, \B_i) = 2$
        \item $b(V_i^+, \S_i) = 1$
        \item $b(V_i^-, \T_i) = 1$
        \item $b(V_i^+, \C_j) = 1$ if $x_i$ is a literal in $C_j$
        \item $b(V_i^-, \C_j) = 1$ if $\neg x_i$ is a literal in $C_j$
    \end{itemize}
    Moreover, the budget of both $V_i^+$ and $V_i^-$ is $1/2$ for all $i \in [n]$.
    \item There is one unbounded buyer $U$ with $b(U, \C_j) = 2$ for all $j \in [m]$ and $b(U, \S_i) = b(U, \T_i) = 2$ for all $i \in [n]$. Moreover, $U$ has a budget of $\infty$.
\end{itemize}

Set the target revenue to be $R = n + m + (3n/2)$. Suppose there exists a $\delta$-approximate throttling equilibrium $\Theta$, for some $\delta \in [0,1)$, with revenue greater than or equal to $R$. Let $\theta_i^+$ and $\theta_i^-$ denote the throttling parameters of $V_i^+$ and $V_i^-$ in $\Theta$. Then, $\theta_i^+ \theta_i^- \leq 1/2$ by virtue of the budget constraints. Therefore, the revenue from goods $\{\A_i\}_{i=1}^n \cup \{\B_i\}_{i=1}^n$ is at most $n$. Furthermore, it is easy to see that the revenue from goods $\{\C_j\}_{j=1}^m$ is at most $m$. Additionally, the total payment by buyer $U$ on goods $\S_i$ and $\T_i$ is at most $\theta_i^+ + \theta_i^- \leq \theta_i^+ + (1/2\theta_i^+)$. Note that $\theta_i^+ + (1/2\theta_i^+)$ is maximized at $\theta_i^+ = 1/2$ or $\theta_i^+ = 1$, with a value of $\theta_i^+ + (1/2\theta_i^+) = 3/2$. Therefore, the revenue from goods $\{\S_i\}_{i=1}^n \cup \{\T_i\}_{i=1}^n$ is at most $3n/2$. Hence, the total payment made on all the goods is at most $R$. 

For the total revenue under $\Theta$ to be greater than or equal to $R$, the revenue from $\{\S_i\}_{i=1}^n \cup \{\T_i\}_{i=1}^n$ must be at least $3n/2$ and the revenue from $\{\C_j\}_{j=1}^m$ must be at least $m$. Hence, under $\Theta$, buyer $U$ has a throttling parameter of $1$, and for each $i \in [n]$, either $(\theta_i^+ = 1, \theta_i^- = 1/2)$ or $(\theta_i^+ = 1/2, \theta_i^- = 1)$. Furthermore, the payment made by buyer $U$ on $\C_j$ is 1 for every $j \in [m]$. This allows us to assign values to the variables as follows: set $x_i = 1$ if $\theta_i^+ = 1$ and $x_i = 0$ if $\theta_i^- = 1$. With this assignment of the variables, each clause is satisfied since the payment made by buyer $U$ on $\C_j$ is 1 for all $j \in [m]$. Hence, we have shown that if there exists a $\delta$-approximate throttling equilibrium with revenue $R$ or greater, then there exists a satisfying assignment for the 3-SAT instance.

Conversely, note that if there exists a satisfying assignment for the 3-SAT instance, then setting $\theta_i^+ = 1$, $\theta_i^- = 1/2$ if $x_i = 1$ and $\theta_i^+ = 1/2$, $\theta_i^- = 1$ if $x_i = 0$ yields a throttling equilibrium with revenue equal to $R$. To complete the proof, observe that the size of the instance $|\I| = \text{poly}(n,m)$.
\end{proof}

%% file: appendix-two-buyer.tex
\subsection{Proof of Theorem~\ref{thm:two-buyer}}\label{appendix:two-buyer}

In this appendix, we analyze the correctness and runtime of Algorithm~\ref{alg:two_buyer}. To do so, we will make repeated use of the following crucial observation:
\begin{align} \label{two_buyer_payment}
    p(\theta)_{ij} =
    \begin{cases}
        \theta_i \theta_k b_{kj} &\text{if } b_{ij} > b_{kj} > 0\ \text{for some } k \in [n]\\
        0 &\text{otherwise}
    \end{cases}
\end{align}
In particular, this observation implies that $p(1, \theta_i)$ is a linear function of $\theta$.

The following lemma makes a step towards the proof of correctness of the algorithm by showing that the budget constraints are always satisfied.

\begin{lemma}\label{lemma:alg_budget}
    At the start of each iteration of the while loop, we have $\sum_j p(\theta)_{ij} \leq B_i$ for all $i \in [n]$.
\end{lemma}
\begin{proof}
    We will use induction on the number of iterations of the while loop to prove this lemma. By our choice of initialization of $\theta$, the budget constraints are satisfied before the first iteration of the while loop. Suppose the constraints are satisfied before the start of the $t$-th iteration and the value of $\theta$ at that stage is $\theta^{(0)}$. We will use $\theta^{(1)}$ and $\theta^{(2)}$ to the denote the value of $\theta$ after step 1 and step 2 of the $t$-th iteration respectively. Consider a buyer $i$ such that $\sum_j p(\theta^{(1)})_{ij} > B_i$. By equation~\ref{two_buyer_payment}, we get $$B_i < \sum_j p(\theta^{(1)})_{ij} \leq \left(\sum_j p(\theta^{(0)})_{ij}\right)/(1 - \gamma)^2$$ which further implies $\sum_j p(\theta^{(0)})_{ij} > (1 - \gamma)^2 B_i$. Therefore, the throttling parameter of buyer $i$ was not changed in step 1 of the $t$-th iteration, i.e., $\theta_i^{(0)} = \theta_i^{(1)}$. As a consequence, we get
    \begin{align*}
        \sum_j p(\theta^{(1)})_{ij} \leq \left(\sum_j p(\theta^{(0)})_{ij}\right)/(1 - \gamma)
    \end{align*}
    After step 2 of the $t$-th iteration, we get $\theta^{(2)} = (1 - \gamma) \theta^{(1)}$. Hence,
    \begin{align*}
        \sum_j p(\theta^{(2)})_{ij} \leq (1 - \gamma)\sum_j p(\theta^{(1)})_{ij} \leq \left(\sum_j p(\theta^{(0)})_{ij}\right) \leq B_i
    \end{align*}
    where the last inequality follows from our inductive hypothesis. As $\theta^{(2)}$ is the value of $\theta$ after the $t$-th iteration, the lemma follows by induction.
\end{proof}

The next lemma establishes that the algorithm never loses any progress, i.e., any buyer who satisfies the `Not too much unnecessary throttling condition' of Definition~\ref{def:throt_eq} at the beginning of some iteration of the while loop continues to do so at the end of it.

\begin{lemma}\label{lemma:alg_progress}
    If $\sum_j p(\theta)_{ij} \geq (1 - \gamma)^3 B_i$ or $\theta_i \geq 1 - \gamma$ at the start of some iteration of the while loop, then $\sum_j p(\theta)_{ij} \geq (1 - \gamma)^3 B_i$ or $\theta_i \geq 1 - \gamma$ at the end of that iteration.
\end{lemma}
\begin{proof}
    Consider an iteration of while loop which starts with $\theta = \theta^{(0)}$. We will use $\theta^{(1)}$ and $\theta^{(2)}$ to the denote the value of $\theta$ after step 1 and step 2 of this iteration. If $\sum_j p(\theta^{(0)})_{ij} \geq (1 - \gamma)^2 B_i$ at the beginning of the iteration, then $\sum_j p(\theta^{(2)})_{ij} \geq (1 - \gamma)^3 B_i$ because
    \begin{align*}
        \sum_j p(\theta^{(1)})_{ij} \geq (1 - \gamma)^2 B_i \quad \text{and} \quad \sum_j p(\theta^{(2)})_{ij} \geq (1 - \gamma) \sum_j p(\theta^{(1)})_{ij}
    \end{align*}
    Suppose $(1 - \gamma)^3 B_i \leq  \sum_j p(\theta^{(0)})_{ij} < (1 - \gamma)^2 B_i$ and $\theta^{(0)}_i < 1 -\gamma$ at the start of the iteration. Then, after step 1, we have $(1 - \gamma)^2 B_i \leq  \sum_j p(\theta^{(1)})_{ij} \leq B_i$. Hence, after step 2, we get $(1 - \gamma)^3 B_i \leq  \sum_j p(\theta^{(3)})_{ij}$. 
    
    Finally, suppose $(1 - \gamma)^3 B_i \leq  \sum_j p(\theta^{(0)})_{ij} < (1 - \gamma)^2 B_i$ and $\theta^{(0)}_i \geq 1 -\gamma$ at the start of the iteration. Then, after step 1, we have $\sum_j p(\theta^{(1)})_{ij} \leq B_i$. Hence, after step 2, we still have $\theta^{(2)}_i \geq (1 - \gamma)$. This completes the proof of the lemma.
\end{proof}

Finally, we combine the above lemmas to establish the correctness and polynomial-runtime of the algorithm.

\begin{proof}[Proof of Theorem~\ref{thm:two-buyer}]
    Let $\theta^*$ be the vector of throttling parameters returned by the algorithm. Lemma~\ref{lemma:alg_budget} implies that $\theta^*$ satisfies the budget constraints of every buyer. Furthermore, upon combining $(1 - \gamma)^3 \geq 1 - 3 \gamma$ with the termination condition of the while loop, we get that either $\theta^*_i \geq 1 - \gamma$ or $\sum_j p(\theta^*)_{ij} \geq (1 - 3\gamma) B_i$ for all $i \in [n]$, which makes $\theta^*$ a $3 \gamma$-approximate throttling equilibrium.
    
    Next, we bound the running time of the algorithm. Define $c = \min_i \min\{B_i/ (2\sum_j b_{ij}), 1\}$. Note that $c \leq \theta_i \leq 1$ for all $i \in [n]$ for the entire run of the algorithm. Based on Lemma~\ref{lemma:alg_progress}, we define
    \begin{align*}
        A(\theta) \coloneqq \{i \in [n] \mid\sum_j p(\theta)_{ij} \geq (1 - \gamma)^3 B_i \text{ or } \theta_i \geq 1 - \gamma
    \end{align*}
    Then Lemma~\ref{lemma:alg_progress} simply states that if $i \in A(\theta)$ at the start of iteration $T$ of the while loop, then $i \in A(\theta)$ at the start of all future iterations $t \geq T$. Moreover, recall that the while loop terminates when $A(\theta) = [n]$. 
    
    Observe that, in each iteration of the while loop, $\theta_i \leftarrow \theta_i/(1- \gamma)$ for some $i \notin A(\theta)$. Hence, the total number of iterations of the while loop $T$ satisfies the following equivalent statements:
    \begin{align*}
        \frac{c}{(1 -\gamma)^{T/n}} \leq 1 \iff T \leq \frac{n \log(1/c)}{\log(1/ (1-\gamma))} \leq \frac{n\log(1/c)}{\gamma}
    \end{align*}
    This completes the proof because each iteration takes polynomially many steps.
\end{proof}

\subsection{Proof of Theorem~\ref{thm:revenue-comparison}}
\begin{proof}[Proof of Theorem~\ref{thm:revenue-comparison}]
    Consider a throttling game $\left(n, m, (b_{ij}), (B_i)\right)$. Let $\theta = (\theta_i)_i$ be the unique throttling equilibrium (TE) and $\alpha = (\alpha_i)_i$ be the unique pacing equilibrium (PE) for this game. We will use $p_j(\theta)$ and $p_j(\alpha)$ to denote the (expected) payment made to the seller on good $j$ under the TE and PE respectively. Then, REV(TE) $= \sum_j p_j(\theta)$ and REV(PE) $= \sum_j p_j(\alpha)$.
    
    First, we show that REV(PE) $\leq 2 \times \text{REV(TE)}$. Let $N \coloneqq \{i \in [n] \mid \theta_i = 1\}$ be the set of buyers who are not budget constrained under the TE. Moreover, define
    \begin{align*}
        M \coloneqq \{j \in [m] \mid \exists\ i\in N \text{ such that } i \text{ wins a non-zero fraction of } j \text{ under the PE } \alpha \}
    \end{align*}
    Note that, since $\theta_i = 1$ and $\alpha_i \leq 1$ for all $i \in N$, we get that the TE yields a higher revenue for the seller on all goods in the set $M$, i.e., $p_j(\theta) \geq p_j(\alpha)$ for all $j \in M$. Therefore, REV(TE) $\geq \sum_{j \in M} p_j(\theta) \geq \sum_{j \in M} p_j(\alpha)$. Furthermore, the definition of throttling equilibrium implies that every buyer $i \notin N$ spends her entire budget $B_i$ under the TE. Hence, by our choice of $M$, we get $\sum_{j \notin M} p_j(\alpha) \leq \sum_{i \notin N} B_i \leq$ REV(TE). Combining the two statements yields REV(PE) $\leq 2 \times \text{REV(TE)}$, as desired.
    
    Next, we complete the proof by showing that REV(TE) $\leq 2 \times \text{REV(PE)}$. Let $S = \{i \in [n] \mid \alpha_i = 1\}$ be the set of buyers who are not budget constrained under the PE. Note that, for all $i \in S$ and $j \in [m]$, buyer $i$ bids $b_{ij}$ under the PE, which implies $p_j(\alpha) \geq \max_{i \in S} b_{ij}$ for all goods $j \in [m]$. Therefore, for any good $j \in [m]$, the total payment made by buyers in the set $S$ under the TE is at most $p_j(\alpha)$. As a consequence, the total payment made by buyers in $S$ under the TE is at most REV(PE). Furthermore, the buyers not in $S$ completely spend their budget under the TE so the total payment made by buyers not in $S$ under the TE is also at most REV(PE). Hence, we have the desired inequality REV(TE) $\leq 2 \times \text{REV(PE)}$.
\end{proof}

%% file: appendix-poa.tex
\subsection{Proof of Theorem~\ref{thm:poa}}\label{appendix:poa}

Fix a throttling equilibrium $\theta \in \Theta$. Recall that we use $X = (X_1, \dots, X_n)$ to capture the random profile of buyers who participate in the auctions, where $X_i = 1$ if and only if buyer $i$ participates in the auctions, and $\Pr(X_i = 1) = \theta_i$. Let $y_{ij}(X)$ be the indicator random variable which equals 1 if and only if good $j$ is allocated to buyer $i$ under the participation profile $X = (X_1, \dots, X_n)$, and is zero otherwise. Moreover, let $p_{j}(X)$ denote the price of item $j$ under the participation profile $X = (X_1, \dots, X_n)$. Here, the price is the highest/second-highest bid for first-price/second-price auctions respectively, and is interpreted to be 0 if no buyers bid in an auction. Observe that
\begin{align*}
    p_{ij}(\theta) = \E\left[ p_{j}(X) y_{ij}(X) \right]\,.
\end{align*}

 Fix a benchmark allocation $y = \{y_{ij}\}$. We begin by establishing the following lemma, which will play a critical role in the proof of the theorem.
\begin{lemma}\label{lemma:poa}
    For all $i \in [n]$, we have
    \begin{align*}
        \min \left\{ \E\left[\sum_{j=1}^m b_{ij} y_{ij}(X) \right], B_i \right\} \geq \min \left\{\sum_{j=1}^m b_{ij} y_{ij}, B_i \right\} - \E\left[ \sum_{j=1}^m p_{j}(X) y_{ij} \right]\,.
    \end{align*}
\end{lemma}
\begin{proof}
    We consider two cases. First assume that $\theta_i < 1$. Then, the no-unnecessary-throttling condition implies that $\sum_{j=1}^m p_{ij}(\theta) = B_i$. Now, observe that $y_{ij}(X) > 0$ only if $b_{ij} \geq p_j(X)$. Consequently, we have
    \begin{align*}
        \E\left[ \sum_{j=1}^m b_{ij} y_{ij}(X)  \right] \geq \E\left[ \sum_{j=1}^m p_j(X) y_{ij}(X)  \right] = \sum_{j=1}^m p_{ij}(\theta) =  B_i\,.
    \end{align*}
    Hence, we get 
    \begin{align*}
        \min \left\{ \E\left[\sum_{j=1}^m b_{ij} y_{ij}(X) \right], B_i \right\} &= B_i\\
        &\geq \min \left\{\sum_{j=1}^m b_{ij} y_{ij}, B_i \right\}\\
        &\geq  \min \left\{\sum_{j=1}^m b_{ij} y_{ij}, B_i \right\} - \E\left[ \sum_{j=1}^m p_{j}(X) y_{ij} \right]\,,
    \end{align*}
    thereby establishing the required lemma statement for a buyer $i$ such that $\theta_i < 1$.

    Next, consider a buyer $i$ such that $\theta_i = 1$, i.e., buyer $i$ always participates. Since $p_j(X) > b_{ij}$ whenever $y_{ij}(X) < 1$, we have
    \begin{align*}
        0 \geq \E\left[ (b_{ij} - p_j(X)) (1 - y_{ij}(X)) y_{ij} \right]\,.
    \end{align*}
    Moreover, we also have
    \begin{align*}
        \E\left[ b_{ij} y_{ij}(X) \right] \geq \E\left[ (b_{ij} - p_j(X)) y_{ij}(X) y_{ij}\right]
    \end{align*}
    Adding the two inequalities, we get
    \begin{align*}
        \E\left[ b_{ij} y_{ij}(X) \right] &\geq \E\left[ (b_{ij} - p_j(X)) (1 - y_{ij}(X)) y_{ij} \right] + \E\left[ (b_{ij} - p_j(X)) y_{ij}(X) y_{ij}\right] = \E\left[ (b_{ij} - p_j(X)) y_{ij} \right]\,.
    \end{align*}
    Summing over all goods $j \in [m]$ yields
    \begin{align*}
        \E\left[ \sum_{j=1}^m b_{ij} y_{ij}(X) \right] &\geq \sum_{j=1}^m b_{ij} y_{ij} - \E\left[ \sum_{j=1}^m p_{j}(X) y_{ij} \right] \\
        &\geq \min \left\{\sum_{j=1}^m b_{ij} y_{ij}, B_i \right\} - \E\left[ \sum_{j=1}^m p_{j}(X) y_{ij} \right]
    \end{align*}
    Additionally, we also have
    \begin{align*}
        B_i \geq \min \left\{\sum_{j=1}^m b_{ij} y_{ij}, B_i \right\} \geq \min \left\{\sum_{j=1}^m b_{ij} y_{ij}, B_i \right\} - \E\left[ \sum_{j=1}^m p_{j}(X) y_{ij} \right]
    \end{align*}
    Therefore,
    \begin{align*}
        \min \left\{ \E\left[\sum_{j=1}^m b_{ij} y_{ij}(X) \right], B_i \right\} \geq \min \left\{\sum_{j=1}^m b_{ij} y_{ij}, B_i \right\} - \E\left[ \sum_{j=1}^m p_{j}(X) y_{ij} \right]\,.
    \end{align*}
    This concludes the lemma by establishing it for buyers $i$ with $\theta_i = 1$.
\end{proof}

With Lemma~\ref{lemma:poa} in hand, we are ready to prove the theorem. First, note that
\begin{align*}
    \sum_{i=1}^m \min \left\{ \E\left[\sum_{j=1}^m b_{ij} y_{ij}(X) \right], B_i \right\} &\geq \sum_{i=1}^m \min \left\{\sum_{j=1}^m b_{ij} y_{ij}, B_i \right\} - \sum_{i=1}^m \E\left[ \sum_{j=1}^m p_{j}(X) y_{ij} \right]\\
    &= \sum_{i=1}^m \min \left\{\sum_{j=1}^m b_{ij} y_{ij}, B_i \right\} - \E\left[ \sum_{j=1}^m p_{j}(X) \sum_{i=1}^m y_{ij} \right]\\
    &= \sum_{i=1}^m \min \left\{\sum_{j=1}^m b_{ij} y_{ij}, B_i \right\} - \E\left[ \sum_{j=1}^m p_{j}(X) \right]\\
    &\geq \sum_{i=1}^m \min \left\{\sum_{j=1}^m b_{ij} y_{ij}, B_i \right\} - \E\left[ \sum_{j=1}^m p_{j}(X) \sum_{i=1}^m y_{ij}(X) \right]\\
    &= \sum_{i=1}^m \min \left\{\sum_{j=1}^m b_{ij} y_{ij}, B_i \right\} - \sum_{i=1}^m \E\left[ \sum_{j=1}^m p_{j}(X)  y_{ij}(X) \right]
\end{align*}
where the second inequality follows from the observation that a good is always allocated whenever it has a positive bid, i.e., $\sum_{i=1}^m y_{ij}(X) = 1$ whenever $p_j(X) > 0$. Hence, if we can show that
\begin{align}\label{eqn:poa-required}
    \sum_{i=1}^m \min \left\{ \E\left[\sum_{j=1}^m b_{ij} y_{ij}(X) \right], B_i \right\} \geq \sum_{i=1}^m \E\left[ \sum_{j=1}^m p_{j}(X)  y_{ij}(X) \right] \,,
\end{align}
we will get
\begin{align*}
    &\sum_{i=1}^m \min \left\{ \E\left[\sum_{j=1}^m b_{ij} y_{ij}(X) \right], B_i \right\} \geq \sum_{i=1}^m \min \left\{\sum_{j=1}^m b_{ij} y_{ij}, B_i \right\} - \sum_{i=1}^m \min \left\{ \E\left[\sum_{j=1}^m b_{ij} y_{ij}(X) \right], B_i \right\} \\
    \iff &\sum_{i=1}^m \min \left\{ 2 \cdot \E\left[\sum_{j=1}^m b_{ij} y_{ij}(X) \right], B_i \right\} \geq  \sum_{i=1}^m \min \left\{\sum_{j=1}^m b_{ij} y_{ij}, B_i \right\}\,
\end{align*}
and thereby complete the proof, because the benchmark allocation $y$ and the throttling equilibrium $\theta$ are both arbitrary. In the remainder, we establish \eqref{eqn:poa-required}.

Since $y_{ij}(X) > 0$ only when $b_{ij} \geq p_j(X)$, we have
\begin{align*}
    \E\left[\sum_{j=1}^m b_{ij} y_{ij}(X) \right] \geq \E\left[\sum_{j=1}^m p_j(X) y_{ij}(X) \right]\,.
\end{align*}
Moreover, the budget constraint of buyer $i$ implies
\begin{align*}
    B_i \geq \E\left[\sum_{j=1}^m p_j(X) y_{ij}(X) \right]\,.
\end{align*}
Combining the two inequalities, we get:
\begin{align*}
    \min \left\{ \E\left[\sum_{j=1}^m b_{ij} y_{ij}(X) \right], B_i \right\} \geq \E\left[ \sum_{j=1}^m p_{j}(X)  y_{ij}(X) \right]\,.
\end{align*}
Summing over all buyers $i \in [n]$ yields \eqref{eqn:poa-required}, as required. Thus we have established the desired bound of $\text{PoA} \leq 2$.

Next, we show that the upper bound on PoA developed above is tight for both first-price and second-price auctions. We do so by demonstrating particular instances for which the bound is tight, starting with the second-price auction format.

\begin{example}\label{example:SP-poa-tight}
    Consider a second-price auction with $m+1$ buyers and $m$ goods for some $m \in \mathbb{Z}_+$. Each of the first $m$ buyers bid $1$ for the $m$ goods respectively and have a budget of $\infty$, i.e., for $i \in [m]$, we have
    \begin{align*}
        b_{ij} =
        \begin{cases}
            1 &\text{if } i = j\\
            0 &\text{if } i \neq j
        \end{cases}\,,
    \end{align*}
    and $B_i = \infty$ (any $B > 2m$ would suffice). The last buyer has bid $b_{m+1,j} = m$ for each of the goods $j \in [m]$ and has a budget of $B_{m+1} = m + \epsilon$ for some small $\epsilon > 0$. In any throttling equilibrium $\theta \in \Theta$, we have $\theta_i = 1$ for all $i \in [m]$ because of the no-unnecessary-throttling condition. Since the sum of all the second-highest bids is $m$ and buyer $m+1$ has the highest bid for every good, she cannot possibly spend her entire budget of $B_{m+1} = m + \epsilon$ and we must also have $\theta_{m+1} = 1$ by the no-unnecessary throttling condition. Therefore, there is a unique throttling equilibrium $\theta$ such that $\theta_i = 1$ for all $i \in [m+1]$ and it has liquid welfare given by
    \begin{align*}
        \lw(y(\theta)) = \left( \sum_{i=1}^m \min\left\{ y_{ii}(\theta), B_i \right\}  \right) + \min\left\{ \sum_{j=1}^m m \cdot y_{m+1,j}(\theta), m + \epsilon \right\} = m+ \epsilon
    \end{align*}
    because $y_{m+1,j}(\theta) = 1$ for all $j \in [m]$. On the other hand, consider the allocation $y$ such that $y_{ii} = 1$ for all $i \in [m-1]$ and $y_{m+1, m} = 1$. It has liquid welfare given by
    \begin{align*}
        \lw(y) = \left( \sum_{i=1}^{m-1} \min\left\{ y_{ii}, B_i \right\}  \right) + \min\left\{ m \cdot y_{m+1,m}, m + \epsilon \right\} = m -1 + m = 2m-1 \,.
    \end{align*}
    Hence, the PoA is at least $(2m-1)/(m+\epsilon)$. As $m$ and $\epsilon$ were arbitrary, we can consider the limit when $m \to \infty$ and $\epsilon \to 0$, which yields the required lower bound of $\text{PoA} \geq 2$.
\end{example}

Observe that in the previous example none of the buyers were throttled ($\theta_i = 1$), which indicates that the lower bound is driven more by the second-price auction format than the specific budget management method, and applies to other methods like pacing. Next, we show that our bound is tight for first-price auctions.

\begin{example}\label{example:FP-poa-tight}
    Consider a first-price auction with $m+1$ buyers and $m+1$ goods, for some $m \in \mathbb{Z}_+$. Each of the first $m$ buyers bid $1$ for the first $m$ goods respectively and bid $m$ on good $m+1$, and have a budget of $1$, i.e., for each $i \in [m]$, we have
    \begin{align*}
        b_{ij} =
        \begin{cases}
            1 &\text{if } j = i\\
            m &\text{if } j = m+1\\
            0 &\text{otherwise}
        \end{cases}\,,
    \end{align*}
    and $B_i = 1$. Moreover, buyer $m+1$ has value $b_{m+1, m+1} = m$ for the $(m+1)$-th good and $b_{m+1, j} = 0$ for all $j \in [m]$, with $B_{m+1} = \infty$.

    Consider a throttling equilibrium $\theta \in \Theta$. We begin by showing that $\theta_i < 1$ for all $i \in [m]$. For contradiction, suppose not. Let $i$ be the smallest index such that $\theta_i = 1$. Then, buyer $i$ spends 1 on good $i$ and spends $m \cdot \prod_{k=1}^{i-1} (1 - \theta_k) > 0$ on good $m+1$ (we use the lexicographic tie-breaking rule), which makes her total expenditure strictly greater than her budget of $B_i = 1$, thereby yielding the required contradiction. Hence, $\theta_i < 1$ for all buyers $i \in [m]$, and consequently, the no-unnecessary-throttling condition implies that their total expected expenditure is exactly 1, i.e., the following equivalent statements hold
    \begin{align}\label{eqn:SP-poa-inter-1}
        \theta_i \cdot 1 + \left( \prod_{k=1}^{i-1} (1 - \theta_k) \right) \cdot \theta_i \cdot m = 1 \quad \iff \quad \theta_i = \frac{1}{1 + \left( \prod_{k=1}^{i-1} (1 - \theta_k) \right) \cdot m} \,.
    \end{align}
    Moreover, since their expenditure is $B_i=1$, that is also their contribution towards the liquid welfare.
    Let $g(i) \coloneqq \prod_{k=1}^{i} (1 - \theta_k)$ denote the probability that the first $i$ buyers do not participate. Next, observe that $\theta_{m+1} =1$ because of the no-unnecessary-throttling condition and $B_{m+1} = \infty$. Therefore, due to the lexicographic tie-breaking rule, buyer $m$ wins good $m+1$ with probability $g(m) = \prod_{k=1}^{m} (1 - \theta_k)$. Hence, the liquid welfare of $\theta$ is given by
    \begin{align*}
        \lw(y(\theta)) = m \cdot 1 + g(m) \cdot m\,
    \end{align*}
    On the other hand, the allocation $y$ which awards good $i$ to buyer $i$ for all $i \in [m+1]$ has $\lw(y) = m + m = 2m$. Consequently, we have
    \begin{align*}
        \text{PoA} \geq \frac{2m}{m + g(m) \cdot m} = \frac{2}{1 + g(m)}\,.
    \end{align*}
    
    To show $\text{PoA} \geq 2$, it suffices to show $\lim_{m \to \infty} g(m) = 0$, which is what we do next. First, observe that \eqref{eqn:SP-poa-inter-1} implies the following recursion for $g(\cdot)$:
    \begin{align*}
        g(i) = (1 - \theta_i) g(i-1) = \frac{\left( \prod_{k=1}^{i-1} (1 - \theta_k) \right) \cdot m}{1 + \left( \prod_{k=1}^{i-1} (1 - \theta_k) \right) \cdot m} \cdot g(i-1) = \frac{g(i-1)^2 \cdot m}{1 + g(i-1) \cdot m}\,.
    \end{align*}
    We will inductively show that $g(i) \leq 1 - i/(m + \sqrt{m})$. Set $b = 1/(m+\sqrt{m})$ The base case $i = 1$ follows because $\theta_1 = 1/(1 + m)$ (see \eqref{eqn:SP-poa-inter-1}). Suppose  $g(i-1) \leq 1 - (i-1)/(m + \sqrt{m})$ for some $i \in [m]$. Then, we have
    \begin{align*}
        g(i) = \frac{g(i-1)^2 \cdot m}{1 + g(i-1) \cdot m} = \frac{m}{\frac{1}{g(i-1)^2} + \frac{m}{g(i-1)}} \leq \frac{m}{\frac{1}{(1 - b(i-1))^2} + \frac{m}{1 - b(i-1)}} = \frac{m \cdot (1 - bi + b)^2}{1 + m (1 - bi + b)} \,.
    \end{align*}
    To complete the induction, it suffices to show:
    \begin{align*}
        & \frac{m \cdot (1 - bi + b)^2}{1 + m (1 - bi + b)} \leq 1 - bi\\
        \iff & m(1 + b^2 i^2 + b^2 -2bi + 2b - 2b^2i) \leq 1 - bi + m(1 + b^2i^2 - 2bi + b - b^2i)\\
        \iff & 1 - bi + m(b^2i - b - b^2) \geq 0\\
        \iff & 1 - m \left( \frac{1}{m + \sqrt{m}} + \frac{1}{(m + \sqrt m)^2} \right) - i \left( \frac{1}{m+\sqrt m} - \frac{m}{(m + \sqrt m)^2} \right) \geq 0\\
        \iff & 1 - m \left( \frac{1}{m + \sqrt{m}} + \frac{1}{(m + \sqrt m)^2} \right) - i \left( \frac{\sqrt m}{(m + \sqrt m)^2} \right) \geq 0\\
    \end{align*}
    To see why the last inequality in the above equivalence chain holds, observe that:
    \begin{align*}
        &1 - m \left( \frac{1}{m + \sqrt{m}} + \frac{1}{(m + \sqrt m)^2} \right) - i \left( \frac{\sqrt m}{(m + \sqrt m)^2} \right)\\
        \geq & 1 - m \left( \frac{1}{m + \sqrt{m}} + \frac{1}{(m + \sqrt m)^2} \right) - m \left( \frac{\sqrt m}{(m + \sqrt m)^2} \right)\\
        = & \frac{(m + \sqrt m)^2 - m(m+\sqrt m) - m - m \sqrt m}{(m + \sqrt m)^2}\\
        = & \frac{m^2 + m + 2 m\sqrt m - m^2 - m \sqrt m - m - m\sqrt m}{(m + \sqrt m)^2}\\
        = & 0
    \end{align*}
    which completes the induction. Hence, $g(m) \leq 1 - m/(m + \sqrt m)$ and $\lim_{m \to \infty} g(m) = 0$, as required.
\end{example}

%% file: appendix-TE-vs-PE.tex
\section{Appendix: Examples for Section~\ref{sec:TE-PE}}\label{appendix:TE-PE}

First, we provide an example to show that the inequality REV(TE) $\leq 2 \times \text{REV(PE)}$ is tight.

\begin{example}
    Consider the throttling game in which there is 1 good and 2 buyers. The bids are given by $b_{11} = 1/\epsilon$, $b_{21} = 1 - \epsilon$ for $\epsilon > 0$ and the budgets are given by $B_1 = 1$, $B_2 = \infty$. Then, in the unique pacing equilibrium, we have $\alpha_1 = \epsilon$ and $\alpha_2 = 1$, whereas in the unique throttling equilibrium, we have $\theta_1 = \epsilon$ and $\theta_2 = 1$. Hence, REV(PE) = 1 and REV(TE) = $1 + (1 - \epsilon)^2$. Since, this is true for arbitrarily small $\epsilon$, we get that the inequality REV(TE) $\leq 2 \times \text{REV(PE)}$ is tight established in Theorem~\ref{thm:revenue-comparison} is tight.
\end{example}

Next, we give a family of examples for which REV(PE) is arbitrarily close to $(4/3) \times \text{REV(TE)}$.

\begin{example}
    Consider a throttling game with 2 goods and 2 buyers. Fix $\epsilon> 0$. The bids are given by $b_{11} = 1 + \epsilon$, $b_{12} = 1$ and $b_{21} = 1$. Moreover, the budgets are given by $B_1 = 1 - \epsilon$ and $B_2 = \infty$. Then, the unique pacing equilibrium is given by $\alpha_1 = 1 - \epsilon$, $\alpha_2 = 1$, and the unique throttling equilibrium is given by $\theta_1 = (1 - \epsilon)/(2 + \epsilon)$, $\theta_2 = 1$. Since $\epsilon$ was arbitrary, we can take it to be arbitrarily small. In which case, we get REV(PE) $\simeq$ 2 and REV(TE) $\simeq$ 1.5, as desired.
\end{example}